\documentclass[a4paper,fleqn]{cas-dc}

\usepackage{amsthm}
\usepackage{mathtools} 

\newtheorem{theorem}{Theorem}

\newtheorem{lemma}[theorem]{Lemma}
\newtheorem{proposition}[theorem]{Proposition}
\newtheorem{conjecture}[theorem]{Conjecture}

\newtheorem{definition}[theorem]{Definition}

\allowdisplaybreaks

\usepackage{import}
\usepackage[capitalise]{cleveref}
\usepackage{subcaption}

\newcolumntype{C}[1]{>{\centering\let\newline\\\arraybackslash\hspace{0pt}}m{#1}}
\usepackage[table]{xcolor}
\usepackage{pdflscape}

\usepackage{tikz}
\usetikzlibrary{arrows,automata}
\usetikzlibrary{positioning}
\usetikzlibrary{calc}
\usetikzlibrary{intersections}
\usetikzlibrary{shapes,fit} 
\usetikzlibrary{math}
\usetikzlibrary {graphs, quotes}
\newcommand{\mini}{\fontsize{4}{5}\selectfont}

\usepackage[colorinlistoftodos,textsize=small]{todonotes}
\usepackage[normalem]{ulem}
\usepackage{pifont}

\usepackage{xstring}

\newcommand{\pardo}[1]{}
\newcommand{\pardone}[1]{}

\usepackage[giveninits=true,maxbibnames=99]{biblatex}
\let\cite\parencite

\makeatletter
\newcommand{\xdashrightarrow}[2][]{\ext@arrow 0359\rightarrowfill@@{#1}{#2}}
\def\rightarrowfill@@{\arrowfill@@\relax\relbar\rightarrow}
\def\leftarrowfill@@{\arrowfill@@\leftarrow\relbar\relax}
\def\leftrightarrowfill@@{\arrowfill@@\leftarrow\relbar\rightarrow}
\def\arrowfill@@#1#2#3#4{%
  $\m@th\thickmuskip0mu\medmuskip\thickmuskip\thinmuskip\thickmuskip
   \relax#4#1
   \xleaders\hbox{$#4#2$}\hfill
   #3$%
}
\makeatother

\title[mode = title]{All ascents exponential from valued constraint graphs of pathwidth three}
\author[1,2]{Artem Kaznatcheev}
\author[1,3]{Willemijn Volgering}

\affiliation[1]{Department of Mathematics, Utrecht University}
\affiliation[2]{Department of Information and Computing Sciences, Utrecht University}
\affiliation[3]{Department of Physics, Utrecht University}

\date{25 February 2026}

\begin{document}

\begin{abstract}
Many combinatorial optimization problems can be formulated as finding an assignment that maximizes some pseudo-Boolean function (that we call the fitness function).
Strict local search starts with some assignment and follows some update rule to proceed to an adjacent assignment of strictly higher fitness.
This means that strict local search algorithms follow ascents in the fitness landscape of the pseudo-Boolean function.
The complexity of the pseudo-Boolean function (and the fitness landscapes that it represents) can be parameterized by properties of the valued constraint satisfaction problem (VCSP) that encodes the pseudo-Boolean function.
We focus on properties of the constraint graphs of the VCSP, with the intuition that spare graphs are less complex than dense ones.
Specifically, we argue that pathwidth is the natural sparsity parameter for understanding limits on the power of strict local search.
We show that prior constructions of sparse VCSPs where all ascents are exponentially long had pathwidth greater than or equal to four.
We improve this this with our controlled doubling construction: 
a valued constraint satisfaction problem of pathwidth three where all ascents are exponentially long from a designated initial assignment.
We conclude that all strict local search algorithms can be forced to take an exponential number of steps even on simple valued constraint graphs of pathwidth three.
\end{abstract}


\maketitle
\thispagestyle{empty} 

\section{Introduction}

\pardo{From combinatorial optimization to fitness landscapes}
Many combinatorial optimization problems can be formulated as finding an assignment $x \in \{0,1\}^V$ where $V$ is some set of variable indexes 
that maximizes some pseudo-Boolean function $f: \{0,1\}^V \rightarrow \mathbb{Z}$~\cite{pseudoBool,QPB} that we call the \emph{fitness function}.
Often, assignments that are `similar' yield similar fitness, so it is natural to introduce an adjacency structure based on just assignments themselves, independent of the particular fitness function.
The simplest such structure is to call two assignments \emph{adjacent} if they differ at a single bit.
A \emph{fitness landscape} is such a pair of pseudo-Boolean function and assignment adjacency structure~\cite{W32,evoPLS,repCP}.

\pardo{local search and ascents}
Local search heuristics are popular methods for navigating such fitness landscapes~\cite{LocalSearchCO_Thesis,LocalSearch_Book1}.
Local search starts with some assignment and follows some update rule to proceed to an adjacent assignment until it eventually finds a good one.
We say the method is a \emph{strict local search} if each update step increases the fitness of the assignment.
Such a method stops at some local peak $x^*$ such that every assignment $y$ adjacent to $x^*$ has the same or lower fitness ($f(y) \leq f(x^*)$).
The resulting sequence $x^0, x^1, \ldots, x^T = x^*$ of step-wise adjacent assignments of increasing fitness is an \emph{ascent} in the fitness landscape $f$~\cite{pw4,pw2MSc,KV2025}.

\pardo{Two modes of failure for local search and PLS-completeness}
There are two common modes of failure for local search.
One well-known mode of failure is getting stuck in a local peak that is not the global peak.
A less well-known mode of failure is taking exponentially long to find any (local) peak -- even in cases where the peak is unique.
\textcite{PLS} introduced the complexity class of polynomial-local search (PLS) to encode the problem of finding \emph{any} local peak (not just the global one) and thus study this second mode of failure for both local and non-local algorithms.
Under this coarse-grained complexity, it is not hard to see that fitness landscapes are PLS-complete even if we restrict to $f$ that are at most quadratic~\cite{W2SAT_PLS1,W2SAT_PLS2,PLS_VCSP2013,PLS_Survey}.
This means that (unless FP = PLS; which is something that most researchers do not believe to be likely) finding \emph{any} local optimum in a general quadratic fitness landscapes is computationally intractable for any polynomial-time algorithm, regardless of whether it relies on local search or not.
We want to build a more nuanced notion of complexity that both teaches us about the structure of easy vs hard fitness landscapes and that is specific to local search algorithms.

\subsection{Fitness landscapes from valued constraints}
\pardo{Parameterized complexity: VCSPs as compact representation of fitness landscapes.
Their constraint graphs.
How to read a constraint graph.
Sparsity and pathwidth.}

Our first step toward a more nuanced notion of complexity 
is to look at subclasses of pseudo-Boolean functions limited by some relevant parameter. 
Any pseudo-Boolean function $f$ can be written as a polynomial:
\begin{equation}
f(x) = \sum_{S \in \mathcal{S}} W(S) \prod_{i \in S}x_i
\end{equation}
\noindent where $\mathcal{S} \subseteq 2^{V}$ is a set of \emph{scopes} and $W: \mathcal{S} \rightarrow \mathbb{Z}$ is a \emph{constraint weight function}.
The triple $\mathcal{C} = (V,\mathcal{S},W)$ is the instance of a \emph{valued constraint satisfaction problem (VCSP)} that represents $f$~\cite{repCP}, and we overload the notation $\mathcal{C}$ to double as the represented pseudo-Boolean function with $\mathcal{C}(x) = f(x)$.
We allow VCSP-instances to be restricted to induced subproblems (generalizing to arbitrary arity from the binary case in \textcite{star_VCSP}):
\begin{definition}[Induced subproblem]
Given a VCSP-instance $\mathcal{C} = (V,\mathcal{S},W)$, any subset $V' \subset V$, and background assignment $z \in \{0,1\}^{V \setminus V'}$,
let $\mathcal{S}' = \{S \cap V' | S \in \mathcal{S}\}$ and $W': \mathcal{S}' \rightarrow \mathbb{Z}$ given by 
\begin{equation}
W'(R)= \sum_{Q \in \mathcal{N}(R,V\setminus V')} W(R \cup Q) \prod_{i \in Q} y_i
\label{eq:subW}
\end{equation}
where $\mathcal{N}(R,V\setminus V') = \{Q \; | \; Q \in V \setminus V' \text{ and } R \cup Q \in \mathcal{S}\}$ is the set of (hyper-edge completion) neighbors of $R$ in $V\setminus V'$,
then $\mathcal{C}' = (V',\mathcal{S'},W')$ is the \emph{induced subproblem} of $\mathcal{C}$ on $V'$ in background $z$.
\end{definition}
\noindent VCSP-subproblems select subcubes of the fitness landscape:
\begin{proposition}
If $\mathcal{C}' = (V',\mathcal{S'},W')$ is an induced subproblem of $\mathcal{C} = (V,\mathcal{S},W)$ on $V'$ in background $z$, we have $\mathcal{C'}(y) = \mathcal{C}(yz)$ for all $y \in \{0,1\}^{V'}$.
\end{proposition}
\begin{proof}
We establish this by showing that the sum over weights of $S \in \mathcal{S}$ in $\mathcal{C}$ (\cref{eq:Cyz} below) and the double sum over weights of $R\cup Q$  for $R \in \mathcal{S}'$ and $Q \in \mathcal{N}(R, V\setminus V')$ in $\mathcal{C}'$ (\cref{eq:C'y} below) are summing over the same terms.
First, notice that for each $S \in \mathcal{S}$ we have the set $S \setminus V' \in \mathcal{N}(S \cap V', V\setminus V')$. 
So by setting $R = S \cap V'$ and $Q = S \setminus V'$, we have at least one $R \cup Q$ in the dobule sum of $\mathcal{C}$' for each $S$ in $\mathcal{C}$.
Second, to see that there is not more than one such $R \cup Q$ for each $S$, notice that for all other $R \in \mathcal{S}'$ if $R \neq S \setminus V'$ then for all $Q \in \mathcal{N}(R,V\setminus V')$ we have that $R \cup Q \neq S$.
This allows us to rewrite $\mathcal{C}$:
\begin{align}
\mathcal{C}(yz) & = \sum_{S \in \mathcal{S}} W(S) \prod_{i \in S \setminus V'} z_i \prod_{j \in S \cap V'} y_j \label{eq:Cyz} \\
& = \sum_{R \in \mathcal{S}'} \underbrace{\Big( \; \quad \smashoperator[lr]{\sum_{Q \in \mathcal{N}(R,V\setminus V')}} W(R \cup Q) \prod_{i \in Q}z_i \Big)}_{W'(R)} \prod_{j \in R} y_j \label{eq:C'y}\\
& = \sum_{R \in \mathcal{S}'} W'(R) \prod_{j \in R} y_j = \mathcal{C}'(y) \qedhere
\end{align}
\end{proof}
For a general pseudo-Boolean function, the corresponding $\mathcal{S}$ will be exponentially large in $|V|$.
So if we want the number of variables ($|V|$) to be the `size' parameter for our complexity measures then it makes sense to restrict to subclasses of pseudo-Boolean functions that are implemented by sets $\mathcal{S}$ that contain only a polynomial number of scopes.
This can be done by, for example, limiting the \emph{arity} of scopes to at most a constant $k$ (i.e., $|S| \leq k$ for $S \in \mathcal{S}$).
We will be especially interested in pseudo-Boolean functions that are affine (i.e., at most unary scopes: $|S| \leq 1$),
quadratic (i.e., at most binary scopes: $|S| \leq 2$), and 
cubic (i.e., at most ternary scopes: $|S| \leq 3$).
We will define the \emph{(valued) constraint graph} of $f$ as the hypergraph with vertexes $V$, hyper-edges given by $\mathcal{S}$ with $W(S) \neq 0$, and labeled by the weight $W(S)$ of the constraint with corresponding scope.
For example:
\begin{equation}
\begin{tikzpicture}
    \node[draw, circle, label = 100:{$-1$}, minimum size =0.8cm] (1) at (-2,0) {$1$};
    \node[draw, circle, label = {$-3$}, minimum size =0.8cm] (2) at (0,0) {$2$};
    \node[draw, circle, label = 70:{$1$}, minimum size =0.8cm] (3) at (2,0) {$3$};
    
    \draw[draw = white, double = black, very thick] (1) -- (2) node [pos = 0.5, fill = white!, sloped] {$2$};
    \draw[draw = white, double = black, very thick] (2) -- (3) node [pos = 0.5, fill = white!, sloped] {$4$};
\end{tikzpicture}
\end{equation}
\noindent encodes $\mathcal{C}(x_1x_2x_3) = -x_1 - 3x_2 + x_3 + 2x_1x_2 + 4x_2x_3$.
We overload the VCSP-instance $\mathcal{C}$ to also refer to the valued constraint graph.
This allows us to measure the `simplicity' of a pseudo-Boolean function through the (hyper)graph parameters of the VCSP's constraint graph.

\pardo{Definition of pathwidth and why we should care.
$K_n$ minors as obstructions to pathwidth/treewidth.}
We are especially interested in measures of simplicity like maximum degree and pathwidth.
Given a set of variable indexes $V$ (i.e., vertices) and set of scopes $\mathcal{S}$ (i.e., hyperedges), 
a \emph{path decomposition} of $(V,\mathcal{S})$ is a sequence of sets (called bins) $P=\{X_1,X_2,\ldots, X_p\}$ where $X_r\subseteq V$ for $r\in[p]$ with the following three properties:
    (1) every vertex $v\in V$ is in at least one bin $X_r$;
    (2) for every hyperedge $S \in \mathcal{S}$ there exists an $r \in [p]$ such that $S \subseteq X_r$;
    (3) for every vertex $v\in V$, if $v\in X_r\cap X_s$ 
    then $v \in X_\ell$ for all $\ell$ such that $r\leq \ell\leq s$.
The \emph{width} of a path decomposition is one less than the size of the largest bin (i.e., $\max_{r\in [p]}|X_r|-1$).
The \emph{pathwidth} of $(V,\mathcal{S})$ is the minimum possible width of a path decomposition of $(V,\mathcal{S})$.
Pathwidth is a more restrictive variant of the general concept of \emph{treewidth}.
Since paths are a kind of tree, $\textrm{treewidth}(V,\mathcal{S}) \leq \text{pathwidth}(V,\mathcal{S})$.
One common way to certify that a graph cannot have low treewidth is to find a large clique as a minor of the graph.
Specficially, if $(V,\mathcal{S})$ has a $K_m$ minor then $m \leq \textrm{treewidth}(V,\mathcal{S}) \leq \text{pathwidth}(V,\mathcal{S})$.
For more details about pathwidth and treewidth and their role in parameterized complexity, see \textcite{B94,parameterizedBook}.

\subsection{Strict local search and the structure of ascents}
\label{sec:strictLocalSearch}

\pardo{Complexity relative to fixed set of algorithms: Different notions of hardness: some ascent long, steepest ascent long, all ascents long.}

We focus on pathwidth as our parameter of sparsity because it illuminates limits on the power of local search versus non-local algorithms.
Specifically, there exist non-local algorithms for finding a (global) optimum in VCSP-instances of bounded treewidth~\cite{BB73,VCSPsurvey}.
This means that VCSPs of bounded treewidth are not PLS-complete (unless FP = PLS).
But this tractability by non-local methods does not mean that finding a local peak in bounded treewidth VCSPs is tractable for local search.
Thus, our second step toward a more nuanced notion of complexity is to restrict the set of potential algorithms to just (certain kinds of) strict local search.
The efficiency of this restricted class of algorithms can then be measured by the length of ascents -- a structural property of fitness landscapes.
We build up this notion and review existing results with a sequence of three questions.

\pardo{All ascents short versus one ascent long:
from smooth landscapes (linear pseudoBoolean functions) to tree-structured VCSPs.
For some ascent long, could start with KM~\cite{KM72} and show that is not sparse~\cite{arityBSc} before moving on to \textcite{repCP}.}
The first question to consider: is there \emph{any} strict local search algorithm that fails on a specific subclass of fitness landscapes?
Or the dual question: are there subclasses of fitness landscapes where \emph{all} strict local search algorithms will find a peak in polynomial time?
Under our definitions, this is equivalent to a structural property of fitness landscapes: it is asking when a fitness landscape will have an exponential ascent versus when can we guarantee all ascents are of polynomial length?
For the primaly question, there are historic constructions like the \textcite{KM72} $d$-cube where an ascent visits all $2^d$ assignments on the way to the peak -- but these constructions cannot be represented by sparse VCSPs~\cite{tw7,arityBSc}.
For the dual question, the obvious partial answer is affine pseudo-Boolean functions, also known as \emph{smooth} fitness landscapes.
On smooth landscapes, each variable has a preferred assignment independent of the assignments of the other $d - 1$ variables.
As such, all ascents have length at most $d$.
This class of all ascents short can be pushed further to include some quadratic pseudo-Boolean functions.
For example, \textcite{repCP} showed that if a constraint graph is tree-structured then all ascents have length at most $\binom{d+1}{2}$.
They also showed that this result cannot be further improved by constructing families of fitness landscapes represented by a sparse VCSP with constraint graphs of maximum degree $3$ and pathwidth $2$ that have exponentially long ascents.
But both the \textcite{KM72} and \textcite{repCP} constructions are easily solved by some strict local search methods like greedy local search.

\pardo{Steepest ascent short versus long: start with KM and similar constructions~\cite{fittestHard2,fittestHard3} and winding landscapes~\cite{evoPLS} and switch to sparse graphs}
So, the second question to consider: does a particular frequently-used strict local search algorithm, like greedy local search, fail to be efficient on a specific subclass of fitness landscapes?
Or under our definitions, when will fitness landscapes have an exponential \emph{steepest} (rather than any) ascent?
Here again, there is a number of historic results that have exponentially long steepest ascents in pseudo-Boolean functions that are not represented by sparse VCSPs~\cite{fittestHard2,fittestHard3,evoPLS}.
Similar results also exist for other popular strict local search methods like random improving step~\cite{randomFitter1,randomFitter2} or random facet~\cite{randomFacetBound,conditionallySmooth}.
Recently, there was a series of improvements using similar techniques to construct sparse VCSPs that are hard for greedy local search.
The first result in the series showed a construction of pathwidth $7$~\cite{tw7}, which was improved to pathwidth $4$~\cite{pw4}, and finally to the best possible result of pathwidth $2$ (and maximum degree $3$)~\cite{pw2MSc,KV2025}.
\textcite{conditionallySmooth} also noted that the much older \textcite{hakenSteepest} exponential steepest ascent construction has pathwidth $3$, although \textcite{hakenSteepest} never studied the pathwidth of their own construction.
Finally, \textcite{conditionallySmooth} answered the dual question of finding  class of landscapes (in their case: conditionally-smooth) where many, but not all, local search methods find a peak in polynomial time.

\pardo{Some ascents short versus all ascents long: tight PLS and snake-in-a-box (long path problem~\cite{longpath})}
This brings us to our final question and the focus of this paper: on what fitness landscapes do \emph{all} strict local search algorithms fail?
Under our definitions, we express this as the structural property:
when will a family of fitness landscapes have all ascents exponential from some designated initial assignment?
Since quadratic pseudo-Boolean functions are PLS-complete under \emph{tight} PLS-reductions, we know that there are some families of fitness landscapes and initial assignments such that all ascents are exponentially long~\cite{W2SAT_PLS1,W2SAT_PLS2,PLS_VCSP2013,PLS_Survey}.
But the instances that we get through these reductions could have dense constraint graphs.
There are also historic examples, like \textcite{longpath}'s recursive construction of a family of fitness landscapes with all ascents exponential from a designated initial assignment.
But \textcite{tw7} proved that any VCSP realizing this recursive construction has to be dense.

\pardo{Spareseness of all-ascents long instances: MAX-CUT and degree.
Degree three is tractable~\cite{MAXCUT_deg3}, but degree four has exponential ascents~\cite{MT2010}; construction since improved by \textcite{MS2024}}

As with the case of long steepest ascent, there has been recent progress on all ascents long from sparse VCSPs.
But this progress has focused on bounded degree rather than pathwidth for sparseness
and Max-Cut rather than general VCSPs for representations.
Max-Cut is the restricted class of pseudo-Boolean functions that have the form:
\begin{equation}
f(x) = \sum_{ij \in E(G)} w_{ij} (x_i(1 - x_j) + (1 - x_i)x_j)
\label{eq:max-cut}
\end{equation}
\noindent where $E(G)$ are the edges of a graph with weights $w_{ij}$.
Clearly, \cref{eq:max-cut} can be rewritten as a binary Boolean VCSP-instance with the same constraint graph (thus, same maximum degree).
On the one hand, \textcite{MAXCUT_deg3} showed that if max degree of Max-Cut constraint graph is less than or equal to three then all ascents are short 
(similarly for general binary Boolean VCSPs, but only of degree less than or equal to two; Theorem 5.2 in \cite{KazThesis}).
On the other hand, \textcite{PLS_MAXCUT5} showed that Max-Cut in constraint graphs of max degree five is PLS-complete under tight reductions.
The remaining case of Max-Cut of degree four as PLS-complete or not is an open problem,
but \textcite{MT2010} have provided an explicit construction of a family of Max-Cut instances of degree four where all ascents are exponential from a designated initial assignment.
Very recently, \textcite{MS2024} have provided a much simpler Max-Cut construction of degree four with all ascents exponential from a designated initial assignment.

\subsection{Our contribution}

\pardo{Statement of main result: VCSP of pathwidth three where all ascents are exponentially long from some initial assignment.
Roadmap/section-sign-positing of how we will show this result.}

Given that max degree as a complexity parameter jumps abruptly from easy to PLS-complete, we believe that bounded pathwidth is a more interesting parameter for illuminating limits on the power of strict local search versus non-local algorithms.
In Section~\ref{sec:prior}, 
after reviewing the prior degree four Max-Cut constructions by \textcite{MT2010} and \textcite{MS2024}, 
we show that they both have a $K_5$ minor and thus pathwidth $\geq 4$ (\cref{prop:MT_K5,prop:MS_pw4}).
In Section~\ref{sec:Ternary-Construction}, we provide a new construction of a VCSP with pathwidth three (\cref{prop:pw3}).
In Section~\ref{sec:exp}, we prove that this construction has an initial assignment from which all ascents are exponentially long (\cref{thm:expAsc}).
Finally, in Section~\ref{sec:disc}, we conjecture that there does not exist any pathwidth two example with all ascents long.

\section{Pathwidth of prior Max-Cut constructions}
\label{sec:prior}

\pardo{Explain gadget structure and how it relates to pathwidth}

For both the prior all ascents long results, \textcite{MT2010} and \textcite{MS2024} designed gadgets ($\mathcal{C}^{\text{MT}/\text{MS}}_k$) on a fixed number of variables and constraints that are then connected together in a chain of gadgets $\{\mathcal{C}^{\text{MT}/\text{MS}}_1, \mathcal{C}^{\text{MT}/\text{MS}}_2, \ldots, \mathcal{C}^{\text{MT}/\text{MS}}_m\}$.
Each gadget has a designated set of two `connector' variables: for the $k$th gadget in the chain, the connector variables are the only ones that are in the scope of constraints to variables in the $(k+1)$th gadget.
All other constraints are internal to each gadget.\footnote{
See \textcite{star_VCSP} for more on chain of gadgets:
the long some-ascents/steepest-ascents examples in \cref{sec:strictLocalSearch} are also chains of gadgets.
}
\textcite{MT2010} had 28 variables and 32 internal constraints per gadget $\mathcal{C}^\text{MT}_k$, as seen in \cref{fig:MT2010Gadget};
the `connector variables' are C13 and D13 and they have two connections to the next gadget with C13 of $\mathcal{C}^\text{MT}_{k-1}$ (C0 in \cref{fig:MT2010Gadget}) connecting to C1 of $\mathcal{C}^\text{MT}_k$ and D13 of $\mathcal{C}^\text{MT}_{k - 1}$ (D0 in \cref{fig:MT2010Gadget}) connecting to C10 of $\mathcal{C}^\text{MT}_k$.
\textcite{MS2024} had 8 variables and 7 internal constraints per gadget ($\mathcal{C}_k^\text{MS}$), as seen in \cref{fig:MS2024};
the connector variables are $(k,1)$ and $(k,8)$ and they have a total of six constraints to $\mathcal{C}^\text{MS}_{k + 1}$: $\{(k+1,j),(k,1)\}$ for $j \in \{2,4,6\}$ and $\{(k+1,i),(k,8)\}$ for $i \in \{3,5,7\}$.

\begin{figure*}
    \centering
        \begin{subfigure}{\columnwidth}
        \centering
        \begin{tikzpicture}[scale=1.0, transform shape, every node/.style={minimum size=19pt, inner sep=2 pt, outer sep=0pt}]
    \def\w{m_k} 
    \def\s{s_k} 
    \tikzmath{\xunit = 1.8; \yunit =1.8;}
        \path[use as bounding box] (-1*\xunit,-2.9*\yunit) rectangle (4*\xunit,0.9*\yunit);
        \node[draw, dashed, circle] (C0) at (-1*\xunit,0) {\mini $C0$};
        \node[draw, dashed, circle] (D0) at (-1*\xunit,-0.5*\yunit) {\mini $D0$};
        
        \node[draw, circle] (C1) at (-0.5*\xunit,0) {\mini $C1$};
        \node[draw, circle] (C2) at (-0.5*\xunit,-0.5*\yunit) {\mini$C2$};
        \node[draw, circle] (C3) at (-0.5*\xunit,-1.25*\yunit) {\mini$C3$};
        \node[draw, circle] (C4) at (0*\xunit,-1.25*\yunit) {\mini$C4$};
        \node[draw, circle] (C5) at (0*\xunit,-0.5*\yunit) {\mini$C5$};
        \node[draw, circle] (C6) at (0*\xunit,0) {\mini$C6$};
        \node[draw, circle] (C7) at (0.5*\xunit,0) {\mini$C7$};
        \node[draw, circle] (C8) at (0.75*\xunit,-0.5*\yunit) {\mini$C8$};
        \node[draw, circle] (C9) at (1.0*\xunit,0) {\mini$C9$};
        \node[draw, circle] (C10) at (1.5*\xunit,0) {\mini$C10$};
        \node[draw, circle] (C12) at (2.0*\xunit,0) {\mini$C12$};
        \node[draw, circle] (C13) at (2.5*\xunit,0) {\mini$C13$};

        \node[draw, circle] (A6) at (0.5*\xunit,0.5*\yunit) {\mini$A6$};
        \node[draw, circle] (B6) at (0*\xunit,0.5*\yunit) {\mini$B6$};

        \node[draw, circle] (D10) at (1.5*\xunit,-0.5*\yunit) {\mini$D10$};
        \node[draw, circle] (E10) at (1.5*\xunit,-1.25*\yunit) {\mini$E10$};
        \node[draw, circle] (F10) at (1.5*\xunit,-2.0*\yunit) {\mini$F10$};

        \node[draw, circle] (G10) at (1.5*\xunit,-2.5*\yunit) {\mini$G10$};
        \node[draw, circle] (G11) at (2.0*\xunit,-2.5*\yunit) {\mini$G11$};
    
        \node[draw, circle] (A12) at (1.5*\xunit,0.5*\yunit) {\mini$A12$};
        \node[draw, circle] (B12) at (2.0*\xunit,0.5*\yunit) {\mini$B12$};
        \node[draw, circle] (D12) at (2.0*\xunit,-0.5*\yunit) {\mini$D12$};
        \node[draw, circle] (F12) at (2.0*\xunit,-1.5*\yunit) {\mini$F12$};
        \node[draw, circle] (G12) at (2.0*\xunit,-2.0*\yunit) {\mini$G12$};

        \node[draw, circle] (D13) at (2.5*\xunit,-0.5*\yunit) {\mini$D13$};
        \node[draw, circle] (F13) at (2.5*\xunit,-2.0*\yunit) {\mini$F13$};
        
        \node[draw, circle] (F14) at (2.5*\xunit,-2.5*\yunit) {\mini$F14$};
        \node[draw, circle] (E14) at (2.5*\xunit,-1.25*\yunit) {\mini$E14$};
        
        \draw[-] (C0) -- (C1);
        \draw[-] (C1) -- (C2);
        \draw[-] (C2) -- (C3);
        \draw[-] (C3) -- (C4);
        \draw[-] (C4) -- (C5);
        \draw[-] (C5) -- (C6);
        \draw[-] (C6) -- (C7);
        \draw[-] (C7) -- (C8);
        \draw[-] (C8) -- (C9);
        \draw[-] (C9) -- (C10);
        \draw[-] (C10) -- (C12);
        \draw[-] (C12) -- (C13);
        
        \draw[-] (A6) -- (B6);
        \draw[-] (B6) -- (C6);

        \draw[-] (C10) -- (D10);
        \draw[-] (D10) -- (E10);
        \draw[-] (E10) -- (F10);

        \draw[-] (F10) -- (G10);
        \draw[-] (G10) -- (G11);
        \draw[-] (G11) -- (G12);
        
        \draw[-] (A12) -- (B12);
        \draw[-] (B12) -- (C12);
        \draw[-] (C12) -- (D12);
        \draw[-] (D12) -- (F12);
        \draw[-] (F12) -- (G12);

        \draw[-] (F12) -- (F13);
        \draw[-] (F13) -- (F14);
        \draw[-] (D13) -- (E14);
        
         \draw[-] (G10) .. controls (2.4*\xunit,-3.5*\yunit) and (3.5*\xunit, -2.6*\yunit).. (E14);

        \draw[-] (C6) .. controls (0.75*\xunit, -1.0*\yunit) .. (G12);
        \draw[-] (E10) -- (G12);
        \draw[-] (E14) -- (D10);
        \draw[-] (E14) .. controls (1.5*\xunit,-0.9*\yunit)..(C8);
        \draw[-] (F10) .. controls +(left:35mm) and +(down:15mm).. (C3);
        \draw[-] (D0) .. controls +(down:45mm) and (0.5*\xunit, -1.8*\yunit) ..(C10);

        \draw[-, dashed] (C13) -- (3.0*\xunit,0*\yunit);
        \draw[-, dashed] (D13) -- (3.0*\xunit,-0.5*\yunit);
    \end{tikzpicture}
        \caption{
        }
        \label{fig:MT2010Gadget}
    \end{subfigure}
    \hfill
    \begin{subfigure}{\columnwidth}
        \centering
        \begin{tikzpicture}[scale=0.8, transform shape, every node/.style={minimum size=30pt, inner sep=2 pt, outer sep=0pt}, trim left = -5cm]
    
    \def\w{m_k} 
    \def\s{s_k} 
    \tikzmath{\xunit = 2.3; \yunit = 0.8;}
        \path[use as bounding box] (-3*\xunit,-6*\yunit) rectangle (3*\xunit,4*\yunit);
        \node[draw, dashed, circle] (8) at (1.5*\xunit,-4.0*\yunit) {{\mini $k-1, 8$}};
        \node[draw, dashed, circle] (1) at (1.3*\xunit,3.0*\yunit) {\mini $k-1,1$};

        \node[draw, dashed, circle] (2) at (-1.0*\xunit,2.0*\yunit) {\mini $k+1, 2$};
        \node[draw, dashed, circle] (4) at (-1.0*\xunit,0.0*\yunit) {\mini $k+1,4$};
        \node[draw, dashed, circle] (6) at (-1.0*\xunit,-2.0*\yunit) {\mini $k+1, 6$};
        
        \node[draw, dashed, circle] (3) at (-1.5*\xunit,1.0*\yunit) {\mini $k+1,3$};
        \node[draw, dashed, circle] (5) at (-1.5*\xunit,-1.0*\yunit) {\mini $k+1, 5$};
        \node[draw, dashed, circle] (7) at (-1.5*\xunit,-3.0*\yunit) {\mini $k+1,7$};
        
        \node[draw, circle] (k1) at (0.0*\xunit,3.0*\yunit) {\tiny $k,1$};
        \node[draw, circle] (k2) at (0.5*\xunit,2.0*\yunit) {\tiny $k, 2$};
        \node[draw, circle] (k3) at (0.0*\xunit,1.0*\yunit) {\tiny $k, 3$};  
        \node[draw, circle] (k4) at (0.5*\xunit,0) {\tiny $k,4$};
        \node[draw, circle] (k5) at (0.0*\xunit,-1.0*\yunit) {\tiny $k,5$};
        \node[draw, circle] (k6) at (0.5*\xunit,-2.0*\yunit) {\tiny $k,6$};
        \node[draw, circle] (k7) at (0.0*\xunit,-3.0*\yunit) {\tiny $k,7$};
        \node[draw, circle] (k8) at (0.5*\xunit,-4.0*\yunit) {\tiny $k,8$};
        
        \draw[-] (k1) -- (k2);
        \draw[-] (k2) -- (k3);
        \draw[-] (k3) -- (k4);
        \draw[-] (k4) -- (k5);
        \draw[-] (k5) -- (k6);
        \draw[-] (k6) -- (k7);
        \draw[-] (k7) -- (k8);
        
        \draw[-] (1) -- (k2);
        \draw[-] (1) -- (k4);
        \draw[-] (1) -- (k6);

        \draw[-, dashed] (k3) ..controls +(-1.0*\xunit,-4.3*\yunit) and +(-1.8*\xunit,-4.0*\yunit).. (8);
        \draw[-, dashed] (8) .. controls +(-1.2*\xunit, -2.5*\yunit) and +(-0.8*\xunit, -3.4*\yunit).. (k5);
        \draw[-, dashed] (8) ..controls +(-0.6*\xunit, -1.0*\yunit) and +(down:0.85*\xunit)..(k7);

        \draw[-, dashed] (k1) -- (2);
        \draw[-, dashed] (k1) -- (4);
        \draw[-, dashed] (k1) -- (6);


        \draw[-, dashed] (3) ..controls +(-1.1*\xunit,-4.5*\yunit) and +(-1.9*\xunit,-4.0*\yunit).. (k8);
        \draw[-, dashed] (k8) .. controls +(-1.3*\xunit, -2.5*\yunit) and +(-0.9*\xunit, -3.4*\yunit).. (5);
        \draw[-, dashed] (k8) ..controls +(-0.6*\xunit, -1.0*\yunit) and +(down:0.7*\xunit)..(7);
        
    \end{tikzpicture}
        \caption{}
        \label{fig:MS2024}
    \end{subfigure}
    \caption{(a) Gadget $\mathcal{C}^\text{MT}_k$ of Monien and Tscheuschner construction~\cite{MT2010} ($\mathcal{C}^\text{MT}$) and (b) Gadget $\mathcal{C}^\text{MS}_k$ of Michel and Scott construction~\cite{MS2024} ($\mathcal{C}^\text{MS}$). 
    Dotted edges and vertices illustrate connections to neighboring gadgets and the constraint weights are omitted.
    In (a) node labels refer to grid position in layout of Figure 5 of \cite{MT2010}: 
    letters for rows, numbers for columns.}
    \label{fig:placeholder}
\end{figure*}

For both $\mathcal{C}^\text{MT}$ and $\mathcal{C}^\text{MS}$, the resulting constraint graph in sparse: both max degree and pathwidth are bounded by a constant.
\textcite{MT2010} and \textcite{MS2024}'s goal was to get the degree of their Max-Cut constraint graphs down to four.
But we want to also show that their chain constructions also had very low pathwidth.
For this it is helpful to notice that the chain construction plays nicely with pathwidth/treewidth:

\begin{lemma}
Given a chain construction $\mathcal{C} = \{\mathcal{C}_1, \mathcal{C}_2, \cdots, \mathcal{C}_n\}$ with respective sets of `connector' variables $Y_1,Y_2,\cdots,Y_n$, the width of a path/tree-decomposition of $\mathcal{C}^X$ is upper bounded by the width a path/tree-decomposition $X_1,X_2,\ldots$, $X_p$ of $Y_{k-1} \cup \mathcal{C}_{k}$ with the extra condition that $X_1 \supseteq Y_{k - 1}$ and $X_p \supseteq Y_k$ are distinct leaves that contain the connectors. 
\end{lemma}

\begin{proof}
A gadget's path/tree-decomposition $X_1,X_2, \ldots, X_p$ with the condition that $Y_{k - 1} \subseteq X_1$ and $Y_k \subseteq X_p$ can be extended to a path/tree-decomposition of the whole construction $\mathcal{C}$ by connecting the paths/trees from each gadget's decomposition at the matching leaves.
\end{proof}

This gives us the pathwidth of both prior constructions.

\begin{proposition}
\textcite{MT2010} construction (\cref{fig:MT2010Gadget}) has $\text{treewidth}(\mathcal{C}^\text{MT}) = 4 \leq \text{pathwidth}(\mathcal{C}^\text{MT})\allowbreak \leq 5$.
\label{prop:MT_K5}
\end{proposition}

\begin{proof}
\begin{description}
    \item[$\text{treewidth}(\mathcal{C^\text{MT}}) \leq 4$:] 
    Tree decomposition is a path $\{X_1, X_2, \dots, X_{17}\}$ with a single branch to $X_\ell$ at the tenth node with bins:
    $X_1 = \{C0, C1, C2, C3, D0\}$, 
    $X_2 = \{C3, 4C, 5C, 6C, D0\}$, 
    $X_3 = \{A6, B6, C3, C6,\allowbreak D0\}$, 
    $X_4 = \{C3,C6,C7, C8, D0\}$, 
    $X_5 = \{C3, C6, \allowbreak C8, \allowbreak C10, D0\}$, 
    $X_6 = \{C3, C6, C8, C9, C10\}$, 
    $X_7 = \{C3, C6, C8, C10, E14\}$,
    $X_8 =\{C3, C6, C10, E14, \allowbreak G12\}$, 
    $X_9 =\{C3, C10, E14, F10, G12\}$, 
    $X_{10} =\{C10, D10, E14, F10, G12\}$, 
    $X_\ell =\{D10, E10, E14, \allowbreak F10, G12\}$, 
    $X_{11} = \{C10, C12, E14, F10, G12\}$,\\
    $X_{12} =\{C12, E14, \allowbreak F10, G10, G12\}$, 
    $X_{13} = \{C12, \allowbreak E14, \allowbreak G10, \allowbreak G11, G12\}$,
    $X_{14} =\{C12, D12, E14, \allowbreak F12, \allowbreak G12\}$, 
    $X_{15} = \{C12, E14, F12, F13, F14\}$,
    $X_{16} =\{A12, \allowbreak B12, \allowbreak C12, C13, F12\}$, 
    $X_{17} =\{C13, D13, F12\}$. 
    \item[$\text{treewidth}(\mathcal{C^\text{MT}}) \geq 4$:]
    $K_5$ minor: $\{C0,D0,C1,C2,C3, \allowbreak C4, \allowbreak C5, \allowbreak C6,B6,6A\}$,
    $\{C10,C9,D10,E10\}$, $\{F10,G10, \allowbreak G11\}$, $\{C7,C8,E14,D13,C13\}$,
    $\{G12,F12,F13, \allowbreak F14, \allowbreak D12,C12,B12,A12\}$.
    \item[$\text{pathwidth}(\mathcal{C^\text{MT}}) \leq 5$:] path decomposition: 
    $\{C0, C1, C2, \allowbreak C3, C4, D0\}$, 
    $\{C3, 4C, 5C, 6C, D0, F10\}$, 
    $\{A6, B6, \allowbreak C6, \allowbreak C10, D0, F10\}$, 
    $\{C6,C7,C8, C9, C10, F10\}$, 
    $\{ \allowbreak C6, \allowbreak C8, \allowbreak C10, E14,F10, G12\}$, 
    $\{C10, E14, F10, G10, \allowbreak G11, \allowbreak G12\}$, 
    $\{C10, E14, F10, G10, G11, G12\}$, 
    $\{C10, \allowbreak D10, \allowbreak E10, E14, F10, G12\}$, 
    $\{C10, E14, F12, \allowbreak F13,$\\$ F14, \allowbreak G12\}$, 
    $\{C10, C12, D12, D13, F12\}$, 
    $\{A12, B12, \allowbreak C12, \allowbreak C13,\allowbreak D13\}$.\qedhere
\end{description}
\end{proof}

\begin{proposition}
    \textcite{MS2024} construction (\cref{fig:MS2024}) has $\text{pathwidth}(\mathcal{C}^\text{MS}) = 4$.
    \label{prop:MS_pw4}
\end{proposition}
\begin{proof}
\begin{description}
\item[$\text{pathwidth}(\mathcal{C^\text{MS}}) \leq 4$:] Path decomposition:\\
$\{(k-1,1),(k-1,8),(k,6),(k,4), (k,5)\}$, 
$\{(k-1,1), \allowbreak (k-1,8),(k,6),(k,4), (k,3)\}$, 
$\{(k-1,1),(k-1,8), \allowbreak (k,6), \allowbreak (k,7), (k,3)\}$, 
$\{(k-1,1),(k,2),(k,4), (k,3)\}$, 
$\{(k,1),(k,8),(k,2), (k,7)\}$.

\item[$\text{pathwidth}(\mathcal{C^\text{MS}}) \geq 4$:] $K_5$ minor: $\{(k-1,1)\}$, $\{(k-1,8)\}$, $\{(k,1),(k,2),(k,3)\}$, $\{(k,4),(k,5)\}$, and \\$\{(k,6),(k,7),(k,8)\}$.\qedhere
\end{description}
\end{proof}

\section{Pathwidth three construction}
\label{sec:Ternary-Construction}

\pardo{The structure of the VCSP instance based on the gadgets with ternary constraints; i.e., defining the variables and scopes of the gadgets.}

To reduce the pathwidth to three, we construct our \emph{controlled doubling} ternary Boolean VCSP-instance $\mathcal{C^\text{CD}_{\leq n}}$ on $8n$ variables as a chain of $n$ gadgets connected by pairs of scopes with each gadget defined on $8$ variables:
\begin{equation}
V_k = \underbrace{\{(k,1),\ldots,(k,6)\}}_\text{`doubling' variables} \cup \underbrace{\{(k,A),(k,B)\}}_\text{`control' variables}
\end{equation}
The gadgets have all eight unary constraints, fourteen binary constraints and one ternary constraints shown in \cref{fig:Gadget-Ter}:
    `doubling' cycles on $(k,1),...,(k,6)$ with six binary scopes $\{(k,1), (k,2)\}$, $\{(k,2), (k,3)\}$, $\{(k,3), (k,6)\}$, $\{(k,1), (k,4)\}$, $\{(k,4), (k,5)\}$, $\{(k,5), (k,6)\}$; 
    constraint between `control' variables $(k,A), \allowbreak (k,B)$ with scope $\{(k,A), \allowbreak  (k,B)\}$; 
    five binary constraints between `control' and `doubling' cycle with scopes $\{(k,1), (k,B)\}$, $\{(k,2), \allowbreak  (k,B)\}$, $\{(k,4),(k,B)\}$, $\{(k,2), (k,A)\}$  and $\{(k,4), \allowbreak  (k,A)\}$; 
    one ternary constraint with scope $\{(k,4),  \allowbreak (k,A), (k,B)\}$; and 
    two binary constraints connecting to neighbouring gadgets with scopes $\{(k,6), (k-1,1)\}$, and $\{(k,B), (k+1,A)\}$.
From this we can check that the controlled doubling construction ($\mathcal{C^\text{CD}}$) has pathwidth $3$:

\begin{figure*}[!tb]
    \centering
    \usetikzlibrary{fit,backgrounds}
\pgfdeclarelayer{middle}
\pgfsetlayers{background,middle,main}

    \begin{tikzpicture}[every node/.style={minimum size=30pt,font=\small}]
    \def\r{r_k}
    \def\s{s_k} 
    \tikzmath{\xunit = 3.1; \yunit =2.8;}

        \begin{pgfonlayer}{main}
        \node[align=center,draw] (M) at (2.6*\xunit,1*\yunit) {$\r=(18n-19)2^k-30(n-k)+8$\\$s_k= 6(n-k)$};
        
        \node[draw, circle, label={[label distance=-0.1cm, rotate=-45]130:$-(2\r+3\s+17)$}] (ANDin) at (-0.9*\xunit,0) {$k,1$};
        \node[draw, circle, label={[label distance=-0.2cm, rotate=20]90:$-(\r+\s+5)$}, fill = white!] (AND1) at (0.*\xunit,\yunit) {$k,2$};
        \node[draw, circle, label={[label distance=-0.2cm, rotate=-20]-90:$-(\r+2\s+9)$}, fill = white!] (AND2) at (0.*\xunit,-\yunit) {$k,4$};

        \node[draw, circle, label={[label distance=-0.2cm,rotate=-20]90:$-(\r+3)$}] (XOR1) at (1.3*\xunit,\yunit) {$k,3$};
        \node[draw, circle, label={[label distance=-0.2cm,rotate=20]-90:$-1$}] (XOR2) at (1.3*\xunit,-\yunit) {$k,5$};
        \node[draw, circle, label={[label distance=-0.1cm,rotate=40]30:$-(\r+1)$}] (XORjoin) at (2.1*\xunit,0) {$k,6$};
        
        \node[draw, circle, label = {[label distance=-0.3cm,rotate=-0]150:$-(\s+1)$}, fill = white!] (B) at (0.*\xunit, 0*\yunit) {$k,B$};
        \node[draw, circle, label = {[label distance=-0.2cm,rotate=-20]90:$-(\s+3)$}, fill = white!] (A) at (1.3*\xunit, 0*\yunit) {$k,A$};
        
        \node[draw, dashed, circle ] (XORout) at (3*\xunit,0) {\tiny $k-1,1$};
        \node[draw, dashed, circle ] (inv) at (-2*\xunit,0) {\tiny $k+1,6$};
        \node[draw, dashed, circle] (inv2) at (-2*\xunit, -1*\yunit) {\tiny $k+1,A$};
        \node[draw, dashed, circle] (K) at (3*\xunit, -1*\yunit) {\tiny $k-1,B$};
        \end{pgfonlayer}

        \begin{pgfonlayer}{background}
        \draw[-, thick] (A) -- (B) node [midway, fill = white!, sloped] {$\s+2$};
        \draw[-, thick] (A)  -- (AND1) node [pos=0.5, fill = white!, sloped] {$-2$};
        \draw[-, thick] (A)  -- (AND2) node [pos=0.5, fill = white!, sloped] {$s_k+2$};
        \draw[-, dashed] (K) ..controls +(left:15mm).. (A) node [pos = 0.65, fill = white!, sloped, above = -0.3] {$\overbrace{\s+6}^{s_{k-1}}$};
        
        \draw[-, thick] (B) -- (AND1) node [pos=0.5, fill = white] {$\s+2$};
        \draw[-, thick] (B)  -- (AND2) node [fill = white!, pos = 0.4] {$2(\s+2)$};
        \draw[-, thick] (B)  -- (ANDin) node [fill = white!,pos = 0.5, sloped] {$-2(\s+2)$};
        \draw[-, thick] (inv2) ..controls +(right:25mm).. (B) node [fill = white!, pos = 0.5, sloped]{$\s$};
        
        \draw[-,dashed] (inv) -- (ANDin) node [midway, fill = white!, above = -0.3] {$\overbrace{2\r+5\s+22}^{r_{k + 1}}$};
        \draw[draw = white, double = black, very thick] (ANDin) -- (AND1) node [pos = 0.4, fill = white!, sloped] {$\r+\s+8$};
        \draw [draw = white, double = black, very thick] (ANDin) -- (AND2)  node [midway, fill = white!, sloped] {$\r++2\s+8$};
        \draw[-, thick] (AND1) -- (XOR1) node [midway, fill = white!, sloped] {$\r+4$};
        \draw[draw = white, double = black, very thick] (AND2) -- (XOR2) node [fill = white!, midway,  sloped] {$\r+4$};
        \draw[-, thick] (XOR1) -- (XORjoin) node [midway, fill = white!, sloped] {$\r+2$};
        \draw[draw = white, double = black, very thick] (XOR2) -- (XORjoin) node [midway, fill = white!, sloped] {$-(\r+2)$};
        \draw[-, thick] (XORjoin) -- (XORout) node [midway, fill = white!, sloped] {$\r$};
        \end{pgfonlayer}
        
        \begin{pgfonlayer}{middle}
        \fill[gray!100, opacity=0.3] (AND2.center) -- (A.center) -- (B.center) -- cycle;

        \coordinate (centroid) at ($(AND1)!0.5!(B)!0.35!(A)$);        
        \coordinate (centroid2) at ($(AND2)!0.65!(B)!0.35!(A)$);

        \node at (centroid2) {\small \underline{$-2(s_k+2)$}};
        \end{pgfonlayer}
        
    \end{tikzpicture}
    \caption{Gadget $\mathcal{C}^{CD}_k$ of controlled doubling construction with $r_k=(18n - 19)2^k - 30(n - k) + 8$ and $s_k = 6(n-k)$:
    unary constraint weights are next to their variables, 
    binary constraint weights are on the edges that specify their scope, 
    ternary constraint weights are in the center of the shaded area that specify their scope. 
    Detailed weights are also given in \cref{sec:weights}.
    Dotted edges and vertices are connections to neighboring gadgets.}
    \label{fig:Gadget-Ter}
\end{figure*}

\pardo{Graph features of the controller doulbing gadget}
\begin{proposition}
The controlled doubling construction (\cref{fig:Gadget-Ter}) has pathwidth $3$.
\label{prop:pw3}
\end{proposition}

\begin{proof}
\begin{description}
\item[$\text{pathwidth}\left(\mathcal{C^\text{CD}}\right) \leq 3$:] Path decomposition:\\
$\{(k,1),(k,B),(k,2), (k,4)\}$,
$\{(k,B),(k,2), \allowbreak (k,4)$,\\ $\allowbreak (k,A)\}$, 
$\{(k,2),(k,4),(k,A), \allowbreak (k,3)\}$, $\{(k,4),(k,A), \allowbreak(k,3), \allowbreak (k,5)\}$,
$\{(k,A),(k,3),(k,5),(k,6)\}$.
\item[$\text{pathwidth}\left(\mathcal{C^\text{CD}}\right) \geq 3$:] $K_4$ minor: $\{(k,1)\}$, $\{(k,A),(k,2), \allowbreak (k,3)\}$, $\{(k,B)\}$, and $\{(k,4),(k,5),(k,6)\}$. \qedhere
\end{description}
\end{proof}

\cref{fig:Gadget-Ter} also gives the internal constraint weights of $\mathcal{C}_{n,k}$ in terms of the parameters $r_k=(18n - 19)2^k - 30(n - k) + 8$ and $s_k = 6(n-k)$,
along with the four constraints that connect $\mathcal{C}_{n,k}$ to $\mathcal{C}_{n,k - 1}$ on the right and $\mathcal{C}_{n,k + 1}$ on the left.
The constraint weights are listed in \cref{sec:weights}.
The weights are set up in such a way that the possible ascents in $\mathcal{C}_{n,k}$ are independent of the value assigned to $x_{(k-1,1)}$ and $x_{(k+1,A)}$.
So of the connector constraints, only the ones with scopes $\{(k,1),(k+1,6)\}$ and $\{(k,A),(k-1,B)\}$ matter for the possible ascents in $\mathcal{C}_{n,k}$.
To see why $(k,6)$ is independent of $(k-1,1)$, let us define:
\begin{multline}
\hat{W}\left((k,6)|x_{(k,3)},x_{(k,5)}\right) \\ = W((k,6)) + \sum_{j \in \{3,5\}}x_{(k,j)}W((k,6),(k,j))
\end{multline}
and notice that:
\begin{multline*}
\hat{W}\left((k,6)|x_{(k,3)},x_{(k,5)}\right) \leq 0 \\
 \Rightarrow |\hat{W}((k,6)|x_{(k,3)},x_{(k,5)})| > W((k,6),(k-1,1)).
\end{multline*}
From this, we see that if $(k,6)$ would want to flip down (i.e., $x_{(k,6)} = 0$ with higher fitness than $x_{(k,6)} = 1$) based on the values of $x_{(k,3)}$ and $x_{(k,5)}$ then the amount of change in fitness would be greater than the contribution of the $\{(k,6),(k-1,1)\}$ constraint.
Thus, the ability of $(k,6)$ to flip does not depend on if $x_{(k-1,1)} = 0$ or $1$.
Same logical applies to $(k,B)$'s independence of $(k+1,A)$, but with $\hat{W}\left((k,B)|x_{(k,1)},x_{(k,2)},x_{(k,4)},x_{(k,A)}\right)$ summing over all the scopes involving $(k,B)$ except for $\{(k,B),(k+1,A)\}$ and the negative magnitude of $\hat{W}$ compared to $W((k,B),(k+1,A))$ (instead of $W((k,6),(k-1,1))$).

This means that although \cref{fig:Gadget-Ter} has four variables outside $V_k$ and thus allows for 16 induced subproblems, if we care about ascents and not exact fitness values then we can consider just the four induced subproblems $\mathcal{C}^{a,b}_{n,k}$ given by assigning to the two background variables $x_{k+1,6} = a$ and $x_{k - 1,B} = b$ (for concreteness, let the two irrelevant variables be set as $x_{k+1,A} = 0$ and $x_{k-1,1} = 0$).

The first gadget ($\mathcal{C}_{n,1}$) does not have another gadget to the right, so the constraints with scopes $\{(1,6), (0,1)\}$ and $\{(0,B), (1,A)\}$ are not possible, so they are replaced by a constraint with scope $\{(1,6),(1,A)\}$ and weight  $W((1,6),\allowbreak(1,A)) = r_1 = s_0 = 6n$.
\cref{fig:Gadget-k1} in Appendix \ref{sec:base_case} shows the constraint weights of $\mathcal{C}_{n,1}$.
Similarly, if our chain construction $\mathcal{C}_{n,\leq m}$ is made of $m \; (\leq n)$ gadgets then the last gadget ($\mathcal{C}_{n,m}$) does not have another gadget ($\mathcal{C}_{n,m+1}$) to the left,
so constraints with scopes $\{(m+1,6),(m,1)\}$ and $\{(m+1,A),(m,B)\}$ are not possible.
We handle this by defining $\mathcal{C}^-_{n,\leq m}$ as the induced subproblem on $V_{\leq m}$ with background $x_{(m+1,6)} = 0$ and $\mathcal{C}^+_{n,\leq m}$ with background $x_{(m+1,6)} = 1$ (for concreteness have both with $x_{(m+1,A)} = 0$, although as outlines above, this value doesn't matter for the assents).
This results in our two chain constructions $\mathcal{C}^\pm_{n,\leq m} = \mathcal{C}^\pm_{n,m}, \mathcal{C}_{n,m-1},\ldots,\mathcal{C}_{n,1}$ that differ only in the unary constraint weight on the $(m,1)$ index, with $W^-((m,1)) = W((m,1)) = -(2r_k + 3s_k + 17)$ and $W^+((m,1)) = W((m,1)) + W((m+1,6),(m,1)) = 2s_k + 5$.

\section{All ascents exponential}
\label{sec:exp}

Now, we have all the tools to prove that our controlled doubling construction has all ascents exponential from a designated initial assignment:

\begin{theorem}
    Consider the VCSP-instances $\mathcal{C}^{+}_{n,\leq m}$ and $\mathcal{C}^{-}_{n,\leq m}$ of the controlled doubling construction from \cref{sec:Ternary-Construction}.
    In the fitness landscape of $\mathcal{C}^{+}_{n,\leq m}$, there is only one ascent from the initial assignment $x^{+}_{\text{start}, \leq m}\coloneq  0^{8m}$ -- this ascent ends at the peak $x^{+}_{\text{end}, \leq m} \coloneq  111110\;01\;0^{8(m-1)}$; 
    and in the fitness landscape of $\mathcal{C}^{-}_{n,\leq m}$, there is only one ascent ascent from the initial assignment $x^{-}_{\text{start}, \leq m} ( = x^{+}_{\text{end},\leq m})$ -- this ascent ends at the peak $x^{-}_{\text{end}, \leq m} ( = x^{+}_{\text{start}, \leq m})$. 
    Both of these ascents flip $x_{(m,B)}$ only on their last step and have a total length $T_{m} = 10(2^m-1)$.
    \label{thm:expAsc}
\end{theorem}

\begin{proof}
\pardo{State induction hypotheses and what we prove in the induction step.}
We will prove that all ascents are exponentially long by induction on the number of gadgets $m$. 

\paragraph{Base case:} 
The only ascent from $000000 \; 00$ to $111110 \; 01$ in the fitness landscapes of $C^+_{n,\leq 1}$ and from $111110 \; 01$ to $000000 \; 00$ in the fitness landscape of $C^-_{n,\leq 1}$ has $10$ steps and flips $x_{(1,B)}$ only in the final step. This follows from Lemma \ref{thm:expAsc-base} in Appendix \ref{sec:base_case}.

\paragraph{Inductive hypothesis:}
For $m - 1$ there exists a:
\begin{enumerate}
\item single ascent starting from $x^{+}_{\text{start}, \leq m-1}\coloneq  0^{8(m-1)}$ in $\mathcal{C}^{+}_{n,\leq m-1}$ to $x^{+}_{\text{end}, \leq m-1} \coloneq 111110 \; 01 \; 0^{8(m-2)}$, and 
\item single ascent starting from $x^{-}_{\text{start}, \leq m-1} \coloneq $\\$ 111110\;01\;0^{8(m-2)}$ in $\mathcal{C}^{-}_{n,\leq m-1}$ to $x^{-}_{\text{end},\leq m-1} \coloneq  0^{8m-1}$
\end{enumerate}
with both ascents having length $T_{m-1}$, and flipping $x_{(B,m-1)}$ only once and as the last step.

\paragraph{Step case for $\mathcal{C}^+_{n,\leq m}$:}
Starting from $x^0 = x^+_{\text{start},m} = 0^{8m}$ means that $x^0_{(m,6)} = 0$ and so the induced subproblem on $V_{\leq m-1}$ in this background is $\mathcal{C}^-_{n,\leq m - 1}$ with initial assignment $x^-_{\text{end},\leq m - 1} = 0^{8(m - 1)}$.
This is a local peak in $\mathcal{C}^-_{n,\leq m - 1}$ and so no flips are possible in $V_{\leq m-1}$ until $x_{(m,6)}$ flips.
This means that the only ascent we can follow is from $y^0 = 000000 \; 00$ in $\mathcal{C}^{+}_{n,m}$ since $x_{(m-1,B)} = 0$ is fixed.
This ascent is given in the first five rows of \cref{tab:ter-ascent}:
\begin{multline}
\hspace{-25pt}
    \footnotesize 
    {\scriptstyle \underline{0}00000\:00}\,x^{-}_{\text{\scriptsize end},\leq m-1} \!\to\! {\scriptstyle \mathbf{1}\underline{0}0000\:00}\,x^{-}_{\text{\scriptsize end},\leq m-1} \!\to\! {\scriptstyle 1\mathbf{1}\underline{0}000\:00} \,x^{-}_{\text{\scriptsize end},\leq m-1} \\  
    \to {\scriptstyle 11\mathbf{1}00\underline{0}\:00}\,x^{-}_{\text{\scriptsize end},\leq m-1} \to {\scriptstyle 11100\mathbf{1}\:00}\,\underline{x}^{+}_{\text{ \scriptsize start}, \leq m-1}.
    \label{eq:block1}
\end{multline}
where the variable that can flip is underlined and the variable that has been flipped is bolded. 
Note that for each of the four steps there is only a single variable that can change its assignment to increase fitness;
and $y^4 = 111001 \: 00$ is a local peak of $\mathcal{C}^+_{n,m}$, so no more flips in $V_m$ are possible.

Further, in the final step of \cref{eq:block1} the assignment of $x_{(m,6)}$ is changed from $0$ to $1$. 
This changes the induced subproblem on $V_{\leq m - 1}$ to $\mathcal{C}^+_{n,\leq m - 1}$ with initial assignment $x^+_{\text{start},\leq m - 1} = 0^{8(m - 1)}$, so the inductive hypothesis applies and we follow the unique long ascent in the fitness landscape corresponding to $\mathcal{C}^+_{n,\leq m - 1}$:
\begin{equation}\hspace{-25pt}
    {\scriptstyle 111001\;00}\;\underline{x}^{+}_{\text{start}, \leq m-1} \xdashrightarrow[\text{\tiny in landscape of } \mathcal{C}^{+}_{n,\leq m-1}]{\text{$T_{m-1}$ {\tiny steps of} $V_{\leq m-1}$ }}
    {\scriptstyle 111001\;\underline{0}0 }\;\mathbf{x}^{+}_{\text{end}, \leq m-1}.
    \label{eq:rec1}
\end{equation}
By the IH, this takes $T_{m-1}$ steps and ends at the local peak of $x^+_{\text{end},\leq m - 1} = 111110 \; 01 \; 0^{8(m-2)}$ of $\mathcal{C}^+_{n,\leq m - 1}$, meaning no more flips are possible in $V_{\leq m -1}$. 

Also by IH, only in the final step of \cref{eq:rec1} does variable $x_{(m-1,B)}$ changes from $0$ to $1$, changing the induced subproblem on $V_{m}$ from $\mathcal{C}^{+}_{n,m}$ to $\mathcal{C}^{1,1}_{n,m}$, 
allowing the next four steps that make up the second block of \cref{tab:ter-ascent}:
\begin{multline}
\hspace{-25pt}
   \footnotesize {\scriptstyle 111001\:\underline{0}0}\,x^{+}_{\text{end}, \leq m-1}\!\to\!{\scriptstyle 111\underline{0}01\:\mathbf{1}0}\,x^{+}_{\text{end}, \leq m-1}\!\to\! {\scriptstyle 111\mathbf{1}\underline{0}1\:10}\,x^{+}_{\text{end}, \leq m-1}\\
   \to {\scriptstyle 1111\mathbf{1}\underline{1}\:10} \,x^{+}_{\text{end}, \leq m-1} \to {\scriptstyle 11111\mathbf{0}\:10}\,\underline{x}^{-}_{\text{start}, \leq m-1}.
\end{multline}
Again, there is only one variable that can change its assignment to increase fitness at each step. 
The final assignment is a fitness peak in the landscape of $\mathcal{C}^{1,1}_{n,m}$, and thus no more flips are possible in $V_{m}$.

Also in the last step, variable $x_{(m,6)}$ flips back from $1$ to $0$, changing the induced subproblem on $V_{\leq m - 1}$ from $\mathcal{C}^{+}_{n,\leq m-1}$ to $\mathcal{C}^{-}_{n,\leq m-1}$.
So $x^{+}_{\text{end}, \leq m-1} = x^{-}_{\text{start}, \leq m-1}$ is no longer a peak on $V_{\leq m-1}$, and the IH applies:
\begin{equation*}
    {\scriptstyle 111110\;10}\;\underline{x}^{-}_{\text{start}, \leq m-1} 
    \xdashrightarrow[\text{\tiny in landscape of } \mathcal{C}^{+}_{n,\leq m-1}]{\text{$T_{m-1}$ {\tiny steps of} $V_{\leq m-1}$ }}
    {\scriptstyle 111110\;\underline{1}0}\;\mathbf{x}^{-}_{\text{end},\leq m-1}.
\end{equation*}
By the IH, only in the last step does the variable $x_{(m-1,B)}$ changed back from $1$ to $0$.
This changes the induced subproblem on $V_m$ from $\mathcal{C}^{1,1}_{n,\leq m}$ to $\mathcal{C}^{+}_{n,\leq m}$.
In this landscape $x^{1,1}_{\text{start}, m} = 111110\;10$ is no longer a peak.
The ascent finished by taking two more steps (in the third block of \cref{tab:ter-ascent}):
\begin{equation*}
    {\scriptstyle 111110\:\underline{1}0} \,{\textstyle x^{-}_{\text{end}, \leq m-1}} \!\to\! {\scriptstyle 111110\:\mathbf{1}\underline{0}}\,{\textstyle x^{-}_{\text{end},\leq m-1}} \!\to\! \underbrace{{\scriptstyle 111110\:0\mathbf{1}}\,{\textstyle x^{-}_{\text{end},\leq m-1}}}_{x^{+}_{\text{end},\leq m}}.
\end{equation*}
Again these are the only possible steps that can be taken.
Putting this together: during the whole ascent from $0^{8m}$ to $111110 \; 01 \; 0^{8(m-1)}$, there are no choices in the steps and thus this is the only ascent, and $x_{(m,B)} = 0$ is fixed until the last step when it flips to $x_{(m,B)} = 1$. 

\paragraph{Step case for $\mathcal{C}^-_{n,\leq m}$:}
\pardo{Explain possible ascents in text using the information of~\cref{tab:ter-ascent} when adding a gadget with landscape $P=0$, $Q=0$ and finish proof.}
Similarly, the ascent in the landscape of $\mathcal{C}^{-}_{n,\leq m}$ starting from $x^{-}_{\text{start}, \leq m} = 111110\;01\;0^{8(m-1)}$ 
has the following steps:
\begin{multline}
\hspace{-25pt}
    \footnotesize
    {\scriptstyle \underline{1}11110\:01}\,x^{-}_{\text{end}, \leq m-1} \!\to\! {\scriptstyle \mathbf{0}11\underline{1}10\:01}\,x^{-}_{\text{end}, \leq m-1} \!\to\! {\scriptstyle 011\mathbf{0}\underline{1}0\:01}\,x^{-}_{\text{end}, \leq m-1} \\
    \to {\scriptstyle 0110\mathbf{0}\underline{0}\:01} \,x^{-}_{\text{end}, \leq m-1} \to {\scriptstyle 01100\mathbf{1}\:01}\,\underline{x}^{+}_{\text{start}, \leq m-1}
\end{multline}
\begin{equation}
    {\scriptstyle 011001\;01}\;\underline{x}^{+}_{\text{start}, \leq m-1} 
    \!\xdashrightarrow[\text{\tiny in landscape of } \mathcal{C}^{+}_{n,\leq m-1}]{\text{$T_{m-1}$ {\tiny steps of} $V_{\leq m-1}$ }}\!
    {\scriptstyle 0\underline{1}1001\;01}\;\mathbf{x}^{+}_{\text{end}, \leq m-1} . 
\end{equation}
\begin{multline}
\hspace{-25pt}
    \footnotesize
    {\scriptstyle 011001\:\underline{0}1}\,x^{+}_{\text{end}, \leq m-1}\!\to\! {\scriptstyle 0\underline{1}1001\:\mathbf{1}1}\,x^{+}_{\text{end}, \leq m-1} \!\to\! {\scriptstyle 0\mathbf{0}\underline{1}001\:11}\,x^{+}_{\text{end}, \leq m-1} \\
    \to {\scriptstyle 00\mathbf{0}00\underline{1}\:11}\,x^{+}_{\text{end}, \leq m-1} \to {\scriptstyle 00000\mathbf{0}\:11}\,\underline{x}^{-}_{\text{start}, \leq m-1} 
\end{multline}
\begin{equation}
    {\scriptstyle 000000\;11}\;\underline{x}^{-}_{\text{start}, \leq m-1}      \xdashrightarrow[\text{\tiny in landscape of } \mathcal{C}^{-}_{n,\leq m-1}]{\text{$T_{m-1}$ {\tiny steps of} $V_{\leq m-1}$ }} {\scriptstyle 000000\;11}\;\mathbf{x}^{-}_{\text{end}, \leq m-1} 
\end{equation}
\begin{equation*}
    {\scriptstyle 000000\:\underline{1}1} \;x^{-}_{\text{end}, \leq m-1}\! \to \!{\scriptstyle 000000\:\mathbf{0}\underline{1}}\;x^{-}_{\text{end}, \leq m-1} \!\to  \!\underbrace{{\scriptstyle 000000\:0\mathbf{0}}\;x^{-}_{\text{end}, \leq m-1}}_{x^{-}_{\text{end}, \leq m}}
\end{equation*}
as shown in the last three blocks of \cref{tab:ter-ascent}.
Just as with the prior case for $\mathcal{C}^+_{n,\leq m}$, this ascent is unique and keeps $(m,B)$ fixed until the last step when it flips from $x_{(m,B)} = 1$ to $x_{(m,B)} = 0$.

\paragraph{Length of ascents:} The ascent in $\mathcal{C}^{+}_{n,\leq m}$ starting from $x^{-}_{\text{start}, \leq m} $ and the ascent in $\mathcal{C}^{-}_{n,\leq m}$ starting from $x^{+}_{\text{start}, \leq m}$ both have length $T_m = 4 + T_{m-1} + 4 + T_{m -1} + 2 = 10 + 2T_{m - 1}$. 
Combing this with the base case of $T_1 = 10$, we obtain $T_m = 10(2^m-1)$.
\end{proof}

\section{Discussion}
\label{sec:disc}

\pardo{High level summary linking back to limits of strict local search}
In this paper, we highlighted the limits of local search compared to more general non-local algorithms.
For this, we focused on problems that can be solved efficiently by non-local methods: VCSPs of bounded pathwidth~\cite{BB73,VCSPsurvey}.
We build the `controlled doubling' VCSP of pathwidth three where every ascent is exponentially long from a designated initial assignment.
Thus, we conclude that all strict local search algorithms can be forced to take an exponential number of steps even on simple VCSPs of pathwidth three.
This leaves only treewidth two as an open question for all ascents long:
is there a Boolean VCSP of treewidth two with all ascents exponential from some designed initial assignment?
We believe that the answer is `no' -- our controlled doubling construction cannot be significantly improved:

\begin{conjecture}
There is a polynomial $p$ such that
given any treewidth two Boolean VCSP-instance on $d$ variables, there is an ascent of length $\leq p(d)$ from any initial assignment.
\end{conjecture}

Our controlled doubling construction takes a chain of gadgets that has a long steepest ascent 
and adds a control system that forces only one gadget to be able to `activate' at any given time.
This converts some exponential ascent into the \emph{only} ascent, thus making \emph{all} ascents exponential from the designated initial assignment.
Our approach creates two `channels' of `information flow' through the VCSP-instance.
For the rightflow channel, we see from $\mathcal{C}^{CD}_{k+1}$ to $\mathcal{C}^{CD}_k$ the `doubling' flow from \textcite{KV2025}, which is captured by the exponentially increasing weights of $m_k$.
For the leftflow channel, we see from $\mathcal{C}^{CD}_{k-1}$ to $\mathcal{C}^{CD}_k$ the `control' flow, which is captured by the linearly decreasing weights of $s_k$.
This keeps only a single gadget at a time capable of flips (whereas for \textcite{KV2025}, many gadgets could have flips but steepness selected flips to be in the correct gadget).
Since \textcite{repCP} showed that any tree-structured Boolean VCSP has at most a quadratic ascent, we know that having \emph{some} ascent exponential requires bins of size $3$ in the constraint graph's tree-decomposition.
To block flips in the `wrong' bins and thus block the short ascents, we think that an extra variable will need to be added to each bin that can potentially be `wrong'.
That is why we expect long ascents to not be possible without bins of size four and thus treewidth greater than or equal to three.

\section*{Acknowledgments}

For helpful discussions, AK would like to thank Peter Jeavons and Sofia Vazquez Alferez, 
and WV would like to thank Lucas Veenema.
We would like to both thank Dave Cohen for helpful disucissions and testing our construction in his VCSP simulator~\cite{GitLab}.


\begin{landscape}
\begin{table}[t]
    \scriptsize
    \centering
    \begin{tabular}{|C{0.49cm}|C{0.16cm} C{0.16cm} C{0.16cm} C{0.16cm} C{0.16cm} C{0.16cm} |C{0.21cm} C{0.2cm}| }
        \hline
        P Q & 1 & 2 & 3 & 4 & 5 & 6 & A & B \\
        \hline
        1 0 & \underline{0} & 0 & 0 & 0 & 0 & 0 & 0 & 0 \\
        &  \textbf{1} & \underline{0} & 0 & 0 & 0 & 0 & 0 & 0 \\
        &  1 & \textbf{1} & \underline{0} & 0 & 0 & 0 & 0 & 0 \\
        &  1 & 1 & \textbf{1} & 0 & 0 & \underline{0} & 0 & 0 \\
        &  1 & 1 & 1 & 0 & 0 & \textbf{1} & 0 & 0 \\
        \hline
        1  1 & 1 & 1 & 1 & 0 & 0 & 1 & \underline{0} & 0\\
        &  1 & 1 & 1 & \underline{0} & 0 & 1 & \textbf{1} & 0 \\
        &  1 & 1 & 1 & \textbf{1} & \underline{0} & 1 & 1 & 0 \\
        &  1 & 1 & 1 & 1 & \textbf{1} & \underline{1} & 1 & 0 \\
        &  1 & 1 & 1 & 1 & 1 & \textbf{0} & 1 & 0\\
        \hline
        1  0 & 1 & 1 & 1 & 1 & 1 & 0 & \underline{1} & 0 \\
        &  1 & 1 & 1 & 1 & 1 & 0 & \textbf{0} & \underline{0} \\
        &  1 & 1 & 1 & 1 & 1 & 0 & 0 & \textbf{1} \\
        \hline
        0  0 & \underline{1} & 1 & 1 & 1 & 1 & 0 & 0 & 1 \\
        &  \textbf{0} & 1 & 1 & \underline{1} & 1 & 0 & 0 & 1 \\
        &  0 & 1 & 1 & \textbf{0} & \underline{1} & 0 & 0 & 1 \\
        &  0 & 1 & 1 & 0 & \textbf{0} & \underline{0} & 0 & 1 \\
        &  0 & 1 & 1 & 0 & 0 & \textbf{1} & 0 & 1 \\
        \hline
        0  1 & 0 & 1 & 1 & 0 & 0 & 1 & \underline{0} & 1\\
        &  0 & \underline{1} & 1 & 0 & 0 & 1 & \textbf{1} & 1 \\
        &  0 & \textbf{0} & \underline{1} & 0 & 0 & 1 & 1 & 1 \\
        &  0 & 0 & \textbf{0} & 0 & 0 & \underline{1} & 1 & 1 \\
        &  0 & 0 & 0 & 0 & 0 & \textbf{0} & 1 & 1 \\
        \hline
        0  0 & 0 & 0 & 0 & 0 & 0 & 0 & \underline{1} & 1 \\
        &  0 & 0 & 0 & 0 & 0 & 0 & \textbf{0} & \underline{1}\\
        &  0 & 0 & 0 & 0 & 0 & 0 & 0 & \textbf{0} \\
        \hline
    \end{tabular}
    \begin{tabular}{|C{1.63cm} C{1.6cm} C{1.23cm} C{1.6cm} C{1.23cm} C{1.23cm} |C{1.23cm} C{1.23cm}| }
        \hline
        1 & 2 & 3 & 4 & 5 & 6 & A & B  \\
        \hline
        \cellcolor{green!25}$2s_k+5$ & $-(r_k+s_k+5)$ & $-(r_k+3)$ & $-(r_k+2s_k+9)$ & $-1$ & $-(r_k+1)$ & $-(s_k+3)$ & $-(s_k+1)$ \\
        $-(2s_k+5)$ & \cellcolor{green!25}$3$ & $-(r_k+3)$ & $-1$ & $-1$ & $-(r_k+1)$ & $-(s_k+3)$ & $-(3s_k+5)$ \\
        $-(r_k+3s_k+13)$ & $-3$ & \cellcolor{green!25}$1$ & $-1$ & $-1$ & $-(r_k+1)$ & $-(s_k+5)$ & $-(2s_k+3)$ \\
        $-(r_k+3s_k+13)$ & $-(r_k+7)$ & $-1$ & $-1$ & $-1$ & \cellcolor{green!25}$1$ & $-(s_k+5)$ & $-(2s_k+3)$ \\
        $-(r_k+3s_k+13)$ & $-(r_k+7)$ & $-(r_k+3)$ & $-1$ & $-(r_k+3)$ & $-1$ & $-(s_k+5)$ & $-(2s_k+3)$ \\
        \hline
        $-(r_k+3s_k+9)$ & $-(r_k+7)$ & $-(r_k+3)$ & $-1$ & $-(r_k+3)$ & $-(r_k+1)$ & \cellcolor{green!25}$1$ & $-(2s_k+3)$ \\
        $-(r_k+3s_k+9)$ & $-(r_k+5)$ & $-(r_k+3)$ & \cellcolor{green!25}$s_k+1$ & $-(r_k+3)$ & $-(r_k+1)$ & $-1$ & $-(s_k+1)$ \\
        $-(2r_k+5s_k+21)$ & $-(r_k+5)$ & $-(r_k+3)$ & $-(s_k+1)$ & \cellcolor{green!25}$1$ & $-(r_k+1)$ & $-(s_k+3)$ & $-(s_k+1)$ \\
        $-(2r_k+5s_k+21)$ & $-(r_k+5)$ & $-(r_k+3)$ & $-(r_k+s_k+5)$ & $-1$ & \cellcolor{green!25}$1$ & $-(s_k+3)$ & $-(s_k+1)$ \\
        $-(2r_k+5s_k+21)$ & $-(r_k+5)$ & $-1$ & $-(r_k+s_k+5)$ & $-(r_k+3)$ & $-1$ & $-(s_k+3)$ & $-(s_k+1)$ \\
        \hline
        $-(2r_k+5s_k+21)$ & $-(r_k+5)$ & $-1$ & $-(r_k+s_k+5)$ & $-(r_k+3)$ & $-(r_k+1)$ & \cellcolor{green!25}$3$ & $-(s_k+1)$ \\
        $-(2r_k+5s_k+21)$ & $-(r_k+7)$ & $-1$ & $-(r_k+3)$ & $-(r_k+3)$ & $-(r_k+1)$ & $-3$ & \cellcolor{green!25}$1$ \\
        $-(2r_k+3s_k+17)$ & $-(r_k+s_k+9)$ & $-1$ & $-(r_k+2s_k+7)$ & $-(r_k+3)$ & $-(r_k+1)$ & $-(s_k+5) $ & $-1$\\
        \hline
        \cellcolor{green!25}$2s_k+5$ & $-(r_k+s_k+9)$ & $-1$ & $-(r_k+2s_k+7)$ & $-(r_k+3)$ & $-(r_k+1)$ & $-(s_k+5)$ & $-(s_k+1)$ \\
        $-(2s_k+5)$ & $-1$ & $-1$ & \cellcolor{green!25}$1$ & $-(r_k+3)$ & $-(r_k+1)$ & $-(s_k+5)$ & $-(3s_k+5)$ \\
        $-(r_k+4s_k+13)$ & $-1$ & $-1$ & $-1$ & \cellcolor{green!25}$1$ & $-(r_k+1)$ & $-3$ & $-(s_k+1)$ \\
        $-(r_k+4s_k+13)$ & $-1$ & $-1$ & $-(r_k+5)$ & $-1$ & \cellcolor{green!25}$1$ & $-3$ & $-(s_k+1)$ \\
        $-(r_k+4s_k+13)$ & $-1$ & $-(r_k+3)$ & $-(r_k+5)$ & $-(r_k+3)$ & $-1$ & $-3$ & $-(s_k+1)$ \\
        \hline
        $-(r_k+4s_k+13)$ & $-1$ & $-(r_k+3)$ & $-(r_k+5)$ & $-(r_k+3)$ & $-(r_k+1)$ & \cellcolor{green!25}$s_k+3$ & $-(s_k+1)$ \\
        $-(r_k+4s_k+13)$ & \cellcolor{green!25}$1$ & $-(r_k+3)$ & $-(r_k+s_k+7)$ & $-(r_k+3)$ & $-(r_k+1)$ & $-(s_k+3)$ & $-(2s_k+3)$ \\
        $-(2r_k+5s_k+21)$ & $-1$ & \cellcolor{green!25}$1$ & $-(r_k+s_k+7)$ & $-(r_k+3)$ & $-(r_k+1)$ & $-(s_k+5)$ & $-(s_k+1)$ \\
        $-(2r_k+5s_k+21)$ & $-(r_k+5)$ & $-1$ & $-(r_k+s_k+7)$ & $-(r_k+3)$ & \cellcolor{green!25}$1$ & $-(s_k+5)$ & $-(s_k+1)$ \\
        $-(2r_k+5s_k+21)$ & $-(r_k+5)$ & $-(r_k+3)$ &$-(r_k+s_k+7)$ & $-1$  & $-1$ & $-(s_k+5)$ &$-(s_k+1)$ \\
        \hline
        $-(2r_k+5s_k+21)$ & $-(r_k+5)$ & $-(r_k+3)$ &$-(r_k+s_k+7)$ & $-1$  & $-(r_k+1)$ & \cellcolor{green!25}$1$ &$-(s_k+1)$ \\
        $-(2r_k+5s_k+21)$ & $-(r_k+3)$ & $-(r_k+3)$ &$-(r_k+5)$ & $-1$  & $-(r_k+1)$ & $-1$ &\cellcolor{green!25}$1$ \\
        $-(2r_k+3s_k+17)$ & $-(r_k+s_k+5)$ & $-(r_k+3)$ &$-(r_k+2s_k+9)$ & $-1$  & $-(r_k+1)$ & $-(s_k+3)$ &$-1$ \\
        \hline  
    \end{tabular}
    \caption{
    \textbf{Important ascents in the gadget $\mathcal{C}^{CD}_k$ as shown in Figure~\ref{fig:Gadget-Ter}.} 
    The columns corresponds to the variable of the $k$'th gadget. 
    The columns are structured in three block types.
    The block with columns $P$ and $Q$ correspond to state of the gadgets to the left and right respectively. 
    The blocks with variables labeled by numbers, correspond to the doubling variables. 
    The blocks with variables labeled by $A$ and $B$, correspond to the controlling variables.
    In the left three blocks the rows correspond to an assignment of the variables.
    The right two blocks show the change in fitness $\Delta \mathcal{C}$ when the assignment of the variable in the corresponding column is changed, while keeping the other variables fixed in the assignment in the left part.  
    The rows are structured based on the assignment of neighboring gadgets, i.e. the assignment of the variables $P$ and $Q$. 
    A positive change in fitness is marked in green.
    The variable that will flip in the ascent is underlined and the variable that has been flipped with respect to the previous assignment is bolded.
    }
    \label{tab:ter-ascent}
\end{table}
\end{landscape}

\appendix
\onecolumn

\section{Weights for controlled doubling construction}
\label{sec:weights}

We give the constraints weights in terms of parameters $r_k = (18n-19)2^k-30(n-k)+8$ and $s_k=6(n-k)$ and show how they relate to each other.
The constraints are grouped by the constraints involving control variables (\cref{eq:c-start}-\eqref{eq:c-end}) and the constraints for the doubling variables (\cref{eq:d-start}-\eqref{eq:d-end}).

\begin{alignat}{2}
    W((k,B), (k+1,A)) & &&= s_k \label{eq:c-start}\\
    W((k,B)) &= -(|W((k,B), (k+1,A))| +1) &&= -(s_k + 1)\\
    W((k,2),(k,B)) &= |W((k,B))| + 1 &&= s_k + 2\\
    W((k,A),(k,B)) &= |W((k,B))| + 1 &&= s_k + 2\\
    W((k,1),(k,B)) & -|W((k,2),(k,B))+W((k,A),(k,B))| &&= -2(s_k+2)\\
    W((k,4),(k,B)) &= |W((k,1),(k,B))|  &&= 2(s_k + 2)\\    
    W((k,2),(k,A),(k,B)) &= -|W((k,4),(k,B))| &&=-2(s_k+2) \label{eq:t-start}\\ 
    W((k,A)) &= -(|W((k,A),(k,B))| + 1) &&= -(s_k+3) \label{eq:CkA}\\
    W((k,2),(k,A)) & &&= -2\\
    W((k,4),(k,A)) &=  &&= s_k + 2\\ 
    W((k-1,B),(k,A)) &= |W((k,A)) + W((k,2),(k,A))| + 1 &&= \underbrace{s_k + 6}_{s_{k-1}} \label{eq:c-end}\\
    W((k,6), (k-1,1)) & &&= r_k \label{eq:d-start}\\
    W((k,6)) &= -(|W((k,6), (k-1,1))| + 1) &&= -(r_k+1)\\
    W((k,3),(k,6)) &= |W((k,6))| + 1  &&= r_k+2\\
    W((k,5),(k,6)) &= -|W((k,3),(k,6))| &&= -(r_k+2)\\
    W((k,3)) &= -(|W((k,3),(k,6))| + 1) &&= -(r_k+3)\\
    W((k,2),(k,3)) &= |W((k,3))| + 1 &&= r_k+4\\
    W((k,2)) &= -|W((k,2),(k,3)+W((k,2),(k,B))| + 1 &&= -(r_k+s_k+5)\\
    W((k,1),(k,2)) &= |W((k,2))+W((k,2),(k,A))| + 1 &&= r_k+6\\
    W((k,5)) & &&= -1\\
    W((k,4),(k,5)) &= |W((k,5)) + W((k,5),(k,6))| + 1 &&= r_k+4\\    
    W((k,4)) &= -(|W((k,4),(k,5)) + W((k,4),(k,B))| + 1) &&= -(r_k+2s_k+9)\\
    W((k,1),(k,4)) &= |W((k,4))| - 1  &&= r_k + 2s_k+8\\
    W((k,1)) &= -(|W((k,1),(k,2)) + W((k,1),(k,4))| + 1) &&= -(2m_k+13) \label{eq:Ck1}\\
    W((k+1,6), (k,1)) &= |W((k,1) + W((k,1),(k,B))| + 1 &&= \underbrace{2r_k+5s_k+22}_{r_{k+1}} \label{eq:d-end}
\end{alignat}

\newpage
\section{Base case: $\mathcal{C}^{\pm}_{n,\leq 1}$}
\label{sec:base_case}

\pardo{Why is edge {(1,6), (1,A)} needed}
For the first gadget there is no gadget to the right. 
Therefor, the constraints with scopes $\{(1,6), (0,1)\}$ and $\{(0,B), (1,A)\}$ are not possible.
In the other gadgets, these constraints ensure that flipping $(k,6)$ eventually leads to a flip of $(k,A)$ via the neighboring gadget $k-1$ on the right.
To replicate this effect, we replace these constraints by binary constraint with scope $\{(1,6), (1,A)\}$ and weight 
\begin{equation}
W((1,6),(1,A))= r_1= s_0=6n
\end{equation}
as shown in \cref{fig:Gadget-k1}.
This constraint enforces that $(1,A)$ flips directly after $(1,6)$ flipped.
%
\pardo{state lemma and prove it using table 2}
Now, we are able to prove that the base case $\mathcal{C}^{\pm}_{n,\leq 1}$ of our controlled doubling construction has a unique ascent of length 10:
\begin{lemma}
    Consider the VCSP-instances $\mathcal{C}^{+}_{n,\leq 1}$ and $\mathcal{C}^{-}_{n,\leq 1}$ from \cref{fig:Gadget-k1}.
    In the fitness landscape of $\mathcal{C}^{+}_{n,\leq 1}$, there is only one ascent from the initial assignment $x^{+}_{\text{start}}\coloneq  000000\;00$ -- this ascent ends at the peak $x^{+}_{\text{end}} \coloneq  111110\;01$; 
    and in the fitness landscape of $\mathcal{C}^{-}_{n,\leq 1}$, there is only one ascent ascent from the initial assignment $x^{-}_{\text{start}} ( = x^{+}_{\text{end}})$ -- this ascent ends at the peak $x^{-}_{\text{end},} ( = x^{+}_{\text{start}})$. 
    Both of these ascents flip $x_{(m,B)}$ only on their last step and have a total length $T_{1} = 10$.
    \label{thm:expAsc-base} 
\end{lemma}
\begin{proof}
    \pardo{Proof in the fitness landscape of $\mathcal{C}^{+}_{n, \leq1}$}
    The ascent starting from $x^{+}_{\text{start}} = 000000\;00$ in the fitness landscape of $\mathcal{C}^{+}_{n, \leq1}$, is provided in Table \ref{tab:ter-ascent-base-P1}. 
    Note that for each step there is only a single variable that can change its assignment to increase fitness. 
    First, the variables in the top arm of the cycle, $(1,1)$, $(1,2)$, $(1,3)$ and $(1,6)$, flip.
    This is followed by a flip of the control variable $(1,A)$. 
    Next, the variables in the bottom arm of the cycle, $(1,4)$, $(1,5)$ and $(1,6)$, will flip.
    Afterwards, the control variable $(1,A)$ flips again. 
    Finally, the ascent ends with a flip of the control variable $(1,B)$.
    So, the ascent ends after 10 steps at $x^{+}_{\text{end}} = 111110\;01$ and $x_{(1,B)}$ is fixed until the last step.

    \pardo{Proof in the fitness landscape of $\mathcal{C}^{-}_{n, \leq1}$}
    The ascent starting from $x^{-}_{\text{start}} = x^{+}_{\text{end}} = 111110\;01$ in the fitness landscape of $\mathcal{C}^{-}_{n, \leq1}$ is given in Table \ref{tab:ter-ascent-base-P0}.
    Again, for each step there is only a single variable that can change its assignment to increase fitness. 
    First, the variable in the bottom arm of the cycle $(1,1)$, $(1,4)$, $(1,5)$ and $(1,6)$ flip.
    Afterwards, the control variable $(1,A)$ flips.
    Next, the variables in the top arm of the cycle $(1,2)$, $(1,3)$ and $(1,6)$ will flip, after which the control variable $(1,A)$ flips again.
    The ascent concludes with the flip of control variable $(1,B)$.
    So, the ascent ends after 10 steps at $x^{-}_{\text{end}} = x^{+}_{\text{start}} = 000000\;00$ and $x_{(1,B)}$ is fixed until the last step.
\end{proof}

\begin{figure*}[!t]
    \centering
    \usetikzlibrary{fit,backgrounds}
\pgfdeclarelayer{middle}
\pgfsetlayers{background,middle,main}

    \begin{tikzpicture}[every node/.style={minimum size=30pt,font=\small}]
    \tikzmath{\xunit = 3.1; \yunit =2.8;}

        \begin{pgfonlayer}{main}
        \node[align=center,draw] (M) at (2.6*\xunit,1*\yunit) {$r_1=6n$\\$s_1= 6n-6$};
        \node[draw, circle, label={[label distance=-0.1cm, rotate=-45]130:$-(30n-1)$}] (ANDin) at (-0.9*\xunit,0) {$1,1$};
        \node[draw, circle, label={[label distance=-0.2cm, rotate=20]90:$-(12n-1)$}, fill = white!] (AND1) at (0,\yunit) {$1,2$};
        \node[draw, circle, label={[label distance=-0.2cm, rotate=-20]-90:$-(18n-3)$}, fill = white!] (AND2) at (0,-\yunit) {$1,4$};

        \node[draw, circle, label={[label distance=-0.2cm,rotate=-20]90:$-(6n+3)$}] (XOR1) at (1.3*\xunit,\yunit) {$1,3$};
        \node[draw, circle, label={[label distance=-0.2cm,rotate=20]-90:$-1$}] (XOR2) at (1.3*\xunit,-\yunit) {$1,5$};
        \node[draw, circle, label={[label distance=-0.1cm,rotate=40]30:$-(6n+1)$}] (XORjoin) at (2.1*\xunit,0) {$1,6$};
        
        \node[draw, circle, label = {[label distance=-0.3cm,rotate=-0]150:$-(6n-5)$}, fill = white!] (B) at (0*\xunit, 0*\yunit) {$1,B$};
        \node[draw, circle, label = {[label distance=-0.2cm,rotate=-20]90:$-(6n-3)$}, fill = white!] (A) at (1.3*\xunit, 0*\yunit) {$1,A$};
        
        \node[draw, dashed, circle ] (inv) at (-2*\xunit,0) {\tiny $2,6$};
        \node[draw, dashed, circle] (inv2) at (-2*\xunit, -1*\yunit) {\tiny $2,A$};
        \end{pgfonlayer}

        \begin{pgfonlayer}{background}
        \draw[-, thick] (A) -- (B) node [midway, fill = white!, sloped] {$6n-4$};
        \draw[-, thick] (A)  -- (AND1) node [pos=0.5, fill = white!, sloped] {$-2$};
        \draw[-, thick] (A)  -- (AND2) node [pos=0.5, fill = white!, sloped] {$6n-4$};
        \draw[-, thick] (XORjoin) -- (A) node [pos = 0.5, fill = white!, sloped, above = -0.3] {$\overbrace{6n}^{m_1=s_0}$};
        
        \draw[-, thick] (B) -- (AND1) node [pos=0.5, fill = white] {$6n-4$};
        \draw[-, thick] (B)  -- (AND2) node [fill = white!, pos = 0.5] {$12n-8$};
        \draw[-, thick] (B)  -- (ANDin) node [fill = white!,pos = 0.5, sloped] {$-(12n-8)$};
        \draw[-, thick] (inv2) ..controls +(right:25mm).. (B) node [fill = white!, pos = 0.5, sloped, above = -0.3]{$\overbrace{6n-6}^{s_1}$};
        
        \draw[-,dashed] (inv) -- (ANDin) node [midway, fill = white!, above = -0.3] {$\overbrace{42n-8}^{m_{2}}$};
        \draw[draw = white, double = black, very thick] (ANDin) -- (AND1) node [pos = 0.4, fill = white!, sloped] {$12n+2$};
        \draw [draw = white, double = black, very thick] (ANDin) -- (AND2)  node [midway, fill = white!, sloped] {$18n-3$};
        \draw[-, thick] (AND1) -- (XOR1) node [midway, fill = white!, sloped] {$6n+4$};
        \draw[draw = white, double = black, very thick] (AND2) -- (XOR2) node [fill = white!, midway,  sloped] {$6n+4$};
        \draw[-, thick] (XOR1) -- (XORjoin) node [midway, fill = white!, sloped] {$6n+2$};
        \draw[draw = white, double = black, very thick] (XOR2) -- (XORjoin) node [midway, fill = white!, sloped] {$-(6n+2)$};
        \end{pgfonlayer}
        
        \begin{pgfonlayer}{middle}
        
        \fill[gray!100, opacity=0.3] (AND2.center) -- (A.center) -- (B.center) -- cycle;

        \coordinate (centroid2) at ($(AND2)!0.65!(B)!0.35!(A)$);

        \node at (centroid2) {\small \underline{$-(12n-8)$}};
        \end{pgfonlayer}
        
    \end{tikzpicture}
    \caption{Gadget $\mathcal{C}^{CD}_{1}$ of controlled doubling construction with $r_1=s_0=6n$:
    unary constraint weights are next to their variables, 
    binary constraint weights are on the edges that specify their scope, 
    ternary constraint weights are in the center of the shaded area that specify their scope. 
    Dotted edges and vertices are connections to the second gadget $\mathcal{C}^{CD}_2$.}
    \label{fig:Gadget-k1}
\end{figure*}

\begin{landscape}
\begin{table}[t]
    \scriptsize
    \centering

    \begin{subtable}[t]{\textwidth}
        \begin{tabular}{|C{0.16cm} C{0.16cm} C{0.16cm} C{0.16cm} C{0.16cm} C{0.16cm} |C{0.21cm} C{0.2cm}| }
        \hline
        1 & 2 & 3 & 4 & 5 & 6 & A & B \\
        \hline
        \underline{0} & 0 & 0 & 0 & 0 & 0 & 0 & 0 \\
        \textbf{1} & \underline{0} & 0 & 0 & 0 & 0 & 0 & 0 \\
        1 & \textbf{1} & \underline{0} & 0 & 0 & 0 & 0 & 0 \\
        1 & 1 & \textbf{1} & 0 & 0 & \underline{0} & 0 & 0 \\
        1 & 1 & 1 & 0 & 0 & \textbf{1} & \underline{0} & 0\\
        1 & 1 & 1 & \underline{0} & 0 & 1 & \textbf{1} & 0 \\
        1 & 1 & 1 & \textbf{1} & \underline{0} & 1 & 1 & 0 \\
        1 & 1 & 1 & 1 & \textbf{1} & \underline{1} & 1 & 0 \\
        1 & 1 & 1 & 1 & 1 & \textbf{0} & \underline{1} & 0 \\
        1 & 1 & 1 & 1 & 1 & 0 & \textbf{0} & \underline{0} \\
        1 & 1 & 1 & 1 & 1 & 0 & 0 & \textbf{1} \\
        \hline
    \end{tabular}
    \begin{tabular}{|C{1.63cm} C{1.6cm} C{1.23cm} C{1.6cm} C{1.23cm} C{1.23cm} |C{1.23cm} C{1.23cm}| }
        \hline
        1 & 2 & 3 & 4 & 5 & 6 & A & B  \\
        \hline
        \cellcolor{green!25}$12n-7$ & $-(12n-1)$ & $-(6n+3)$ & $-(18n-3)$ & $-1$ & $-(6n+1)$ & $-(6n-3)$ & $-(6n-5)$ \\
        $-(12n-7)$ & \cellcolor{green!25}$3$ & $-(6n+3)$ & $-1$ & $-1$ & $-(6n+1)$ & $-(6n-3)$ & $-(18n-13)$ \\
        $-(24n-5)$ & $-3$ & \cellcolor{green!25}$1$ & $-1$ & $-1$ & $-(6n+1)$ & $-(6n-6+5)$ & $-(12n-9)$ \\
        $-(24n-5)$ & $-(6n+7)$ & $-1$ & $-1$ & $-1$ & \cellcolor{green!25}$1$ & $-(6n-1)$ & $-(12n-9)$ \\
        $-(24n-9)$ & $-(6n+7)$ & $-(6n+3)$ & $-1$ & $-(6n+3)$ & $-1$ & \cellcolor{green!25}$1$ & $-(12n-9)$ \\
        $-(24n-9)$ & $-(6n+5)$ & $-(6n+3)$ & \cellcolor{green!25}$6n-5$ & $-(6n+3)$ & $-(6n+1)$ & $-1$ & $-(6n-5)$ \\
        $-(42n-9)$ & $-(6n+5)$ & $-(6n+3)$ & $-(6n-6+1)$ & \cellcolor{green!25}$1$ & $-(6n+1)$ & $-(6n-3)$ & $-(6n-5)$ \\
        $-(42n-9)$ & $-(6n+5)$ & $-(6n+3)$ & $-(12n-1)$ & $-1$ & \cellcolor{green!25}$1$ & $-(6n-3)$ & $-(6n-5)$ \\
        $-(42n-9)$ & $-(6n+5)$ & $-1$ & $-(12n-1)$ & $-(6n+3)$ & $-1$ & \cellcolor{green!25}$3$ & $-(6n-5)$ \\
        $-(42n-9)$ & $-(6n+7)$ & $-1$ & $-(6n+3)$ & $-(6n+3)$ & $-(6n+1)$ & $-3$ & \cellcolor{green!25}$1$ \\
        $-(30n-1)$ & $-(12n+3)$ & $-1$ & $-(18n-5)$ & $-(6n+3)$ & $-(6n+1)$ & $-(6n-1) $ & $-1$\\
        \hline
    \end{tabular}
    \caption{$\mathcal{C}^{+}_{n,\leq 1}$}
    \label{tab:ter-ascent-base-P1}
    \end{subtable}
    \begin{subtable}[t]{\textwidth}
        \begin{tabular}{|C{0.16cm} C{0.16cm} C{0.16cm} C{0.16cm} C{0.16cm} C{0.16cm} |C{0.21cm} C{0.2cm}| }
        \hline
        1 & 2 & 3 & 4 & 5 & 6 & A & B  \\
        \hline
        \underline{1} & 1 & 1 & 1 & 1 & 0 & 0 & 1 \\
        \textbf{0} & 1 & 1 & \underline{1} & 1 & 0 & 0 & 1 \\
        0 & 1 & 1 & \textbf{0} & \underline{1} & 0 & 0 & 1 \\
        0 & 1 & 1 & 0 & \textbf{0} & \underline{0} & 0 & 1 \\
        0 & 1 & 1 & 0 & 0 & \textbf{1} & \underline{0} & 1\\
        0 & \underline{1} & 1 & 0 & 0 & 1 & \textbf{1} & 1 \\
        0 & \textbf{0} & \underline{1} & 0 & 0 & 1 & 1 & 1 \\
        0 & 0 & \textbf{0} & 0 & 0 & \underline{1} & 1 & 1 \\
        0 & 0 & 0 & 0 & 0 & \textbf{0} & \underline{1} & 1 \\
        0 & 0 & 0 & 0 & 0 & 0 & \textbf{0} & \underline{1}\\
        0 & 0 & 0 & 0 & 0 & 0 & 0 & \textbf{0} \\
        \hline
    \end{tabular}
    \begin{tabular}{|C{1.63cm} C{1.6cm} C{1.23cm} C{1.6cm} C{1.23cm} C{1.23cm} |C{1.23cm} C{1.23cm}| }
        \hline
        1 & 2 & 3 & 4 & 5 & 6 & A & B  \\
        \hline
        \cellcolor{green!25}$12n-7$ & $-(12n+3)$ & $-1$ & $-(18n-5)$ & $-(6n+3)$ & $-(6n+1)$ & $-(6n-1)$ & $-(6n-5)$ \\
        $-(12n-7)$ & $-1$ & $-1$ & \cellcolor{green!25}$1$ & $-(6n+3)$ & $-(6n+1)$ & $-(6n-1)$ & $-(18n-13)$ \\
        $-(30n-11)$ & $-1$ & $-1$ & $-1$ & \cellcolor{green!25}$1$ & $-(6n+1)$ & $-3$ & $-(6n-5)$ \\
        $-(30n-11)$ & $-1$ & $-1$ & $-(6n+5)$ & $-1$ & \cellcolor{green!25}$1$ & $-3$ & $-(6n-5)$ \\
        $-(30n-11)$ & $-1$ & $-(6n+3)$ & $-(6n+5)$ & $-(6n+3)$ & $-1$ & \cellcolor{green!25}$6n-3$ & $-(6n-5)$ \\
        $-(30n-11)$ & \cellcolor{green!25}$1$ & $-(6n+3)$ & $-(12n+1)$ & $-(6n+3)$ & $-(6n+1)$ & $-(6n-3)$ & $-(12n-9)$ \\
        $-(42n-9)$ & $-1$ & \cellcolor{green!25}$1$ & $-(12n-1)$ & $-(6n+3)$ & $-(6n+1)$ & $-(6n-1)$ & $-(6n-5)$ \\
        $-(42n-9)$ & $-(6n+5)$ & $-1$ & $-(12n+1)$ & $-(6n+3)$ & \cellcolor{green!25}$1$ & $-(6n-1)$ & $-(6n-5)$ \\
        $-(42n-9)$ & $-(6n+5)$ & $-(6n+3)$ &$-(12n+1)$ & $-1$  & $-1$ & \cellcolor{green!25}$1$ &$-(6n-5)$ \\
        $-(42n-9)$ & $-(6n+3)$ & $-(6n+3)$ &$-(6n+5)$ & $-1$  & $-(6n+1)$ & $-1$ &\cellcolor{green!25}$1$ \\
        $-(30n-1)$ & $-(12n-1)$ & $-(6n+3)$ &$-(18n-3)$ & $-1$  & $-(6n+1)$ & $-(6n-3)$ &$-1$ \\
        \hline  
    \end{tabular}
    \caption{$\mathcal{C}^{-}_{n,\leq 1}$}
    \label{tab:ter-ascent-base-P0}
    \end{subtable}
    \caption{
    \textbf{Ascent in the first gadget $\mathcal{C}^{CD}_{1}$ as shown in Figure~\ref{fig:Gadget-k1}.} 
    The columns corresponds to the variable of the first gadget. 
    The columns are structured in two block types.
    The blocks with variables labeled by numbers, correspond to the doubling variables. 
    The blocks with variables labeled by $A$ and $B$, correspond to the controlling variables.
    In the left two blocks the rows correspond to an assignment of the variables.
    The right two blocks show the change in fitness $\Delta \mathcal{C}$ when the assignment of the variable in the corresponding column is changed, while keeping the other variables fixed in the assignment in the left part.  
    A positive change in fitness is marked in green.
    The variable that will flip in the ascent is underlined and the variable that has been flipped with respect to the previous assignment is bolded.
    }
\end{table}
\end{landscape}

\newpage
\printbibliography

\end{document}